\documentclass[a4paper]{article}
\usepackage{harvard}
\usepackage{graphicx}
\usepackage{amssymb}
\usepackage{amsfonts}
\usepackage{amsmath}
\usepackage{amsthm}
\usepackage{geometry}
\usepackage{scalefnt}
\usepackage{numinsec}

\citationmode{abbr}
\harvardparenthesis{none}
\harvardyearparenthesis{round}
\newtheorem{theorem}{Theorem}[section]
\newtheorem*{ack}{Acknowledgement}

\newtheorem{condition}[theorem]{Condition}

\newtheorem{proposition}[theorem]{Proposition}

\input{tcilatex}
\begin{document}

\title{Distributed-order fractional wave equation on a finite domain. Stress
relaxation in a rod}
\author{Teodor M. Atanackovic
\begin{footnote}
{Department of Mechanics, Faculty of Technical Sciences,
University of Novi Sad, Trg D. Obradovica, 6, 21000 Novi Sad,
Serbia, atanackovic@uns.ac.rs}
\end{footnote}, Stevan Pilipovic
\begin{footnote}
{Department of Mathematics, Faculty of Natural Sciences and
Mathematics, University of Novi Sad, Trg D. Obradovica, 3, 21000
Novi Sad, Serbia, stevan.pilipovic@dmi.uns.ac.rs}
\end{footnote} and Dusan Zorica
\begin{footnote}
{Faculty of Civil Engineering, University of Novi Sad, Kozaracka
2a, 24000 Subotica, Serbia, zorica@gf.uns.ac.rs}
\end{footnote}}
\maketitle

\begin{abstract}
\noindent We study waves in a rod of finite length with a viscoelastic
constitutive equation of fractional distributed-order type for the special
choice of weight functions. Prescribing boundary conditions on displacement,
we obtain case corresponding to stress relaxation. In solving system of
differential and integro-differential equations we use the Laplace
transformation in the time domain.

\bigskip

\end{abstract}

\label{firstpage}

\section{Introduction}

Fractional derivatives have been used in describing physical phenomena such
as viscoelasticity, diffusion and wave phenomena. There are two approaches
in formulating differential equations with fractional derivatives in physics
and mechanics. In the first approach classical "integer order" differential
equations of a process are modified by introducing fractional derivatives
instead of integer order ones (see books by \cite{Mai1}, \cite{Pod}, and
\cite{TAFDE}). In the second approach one uses variational principles such
as the Hamilton principle as a starting point for deriving equations of a
process, where a modification of the classical case is achieved by replacing
some (or all) integer order derivatives in Lagrangian density by fractional
derivatives of certain kind. Then the resulting Euler-Lagrange equations are
equations of a process and they contain both left and right fractional
derivatives (see papers by \cite{Agr02}, \cite{A-S}, \cite{AKP}).

In this paper we generalize classical wave equation for one-dimensional
elastic body by following the first approach. Recall the classical setting.
Consider the equation of motion
\begin{equation}
\frac{\partial }{\partial x}\sigma \left( x,t\right) =\rho \frac{\partial
^{2}}{\partial t^{2}}u\left( x,t\right) ,\;\;x\in \left[ 0,L\right] ,\;t>0,
\label{1}
\end{equation}%
where $\rho $, $\sigma $ and $u$ denote density, stress and displacement of
a material at a point positioned at $x$ and at a time $t,$ respectively. It
is coupled with the Hooke Law
\begin{equation}
\sigma \left( x,t\right) =E\mathcal{E}\left( x,t\right) ,\;\;x\in \left[ 0,L%
\right] ,\;t>0,  \label{1-1}
\end{equation}%
where $E$ is a modulus of elasticity and $\mathcal{E}$ is a strain measure,
defined by%
\begin{equation}
\mathcal{E}\left( x,t\right) =\frac{\partial }{\partial x}u\left( x,t\right)
,\;\;x\in \left[ 0,L\right] ,\;t>0.  \label{3}
\end{equation}%
Combining (\ref{1}) - (\ref{3}), the classical wave equation is obtained as
\begin{equation*}
\frac{\partial ^{2}}{\partial x^{2}}u\left( x,t\right) =\frac{\rho }{E}\frac{%
\partial ^{2}}{\partial t^{2}}u\left( x,t\right) ,\;\;x\in \left[ 0,L\right]
,\;t>0.
\end{equation*}

We propose the generalization of a constitutive equation (\ref{1-1}) by
replacing it with a constitutive equation which corresponds to a generalized
viscoelastic body:%
\begin{equation}
\int_{0}^{1}\phi _{1}\left( \alpha \right) {}_{0}D_{t}^{\alpha }\sigma
\left( x,t\right) \mathrm{d}\alpha =E\int_{0}^{1}\phi _{2}\left( \alpha
\right) {}_{0}D_{t}^{\alpha }\mathcal{E}\left( x,t\right) \mathrm{d}\alpha
,\;\;x\in \left[ 0,L\right] ,\;t>0,  \label{3-1}
\end{equation}%
where $E$ is a positive constant (having dimension of stress), $\phi _{1}$
and $\phi _{2}$ are given functions or distributions and $%
{}_{0}D_{t}^{\alpha }y$ is the left Riemann-Liouville fractional derivative
of a function $y\in AC\left( \left[ 0,T\right] \right) ,$ for every $T>0,$
of the order $\alpha \in \left[ 0,1\right) ,$ defined as
\begin{equation*}
_{0}D_{t}^{\alpha }y\left( t\right) :=\frac{1}{\Gamma \left( 1-\alpha
\right) }\frac{\mathrm{d}}{\mathrm{d}t}\int_{0}^{t}\frac{y\left( \tau
\right) }{\left( t-\tau \right) ^{\alpha }}\mathrm{d}\tau ,\;\;t>0,
\end{equation*}%
where $\Gamma $ is the Euler gamma function. Recall, $AC\left( \left[ 0,T%
\right] \right) $ denotes the space of absolutely continuous functions (for
a detailed account on fractional calculus see a book \cite{SKM}). In case
when $\phi _{1}$ and $\phi _{2}$ are distributions$,$ we assume that $\phi
_{1}$ and $\phi _{2}$ are compactly supported by $\left[ 0,1\right] $ ($\phi
_{1},\phi _{2}\in \mathcal{E}^{\prime }\left(
\mathbb{R}
\right) ,$ $\func{supp}\phi _{1},\func{supp}\phi _{2}\subset \left[ 0,1%
\right] $). In this case integrals in (\ref{3-1}) are defined as%
\begin{equation*}
\left\langle \int_{\func{supp}\phi }\phi \left( \alpha \right)
{}_{0}D_{t}^{\alpha }h\left( t\right) \mathrm{d}\alpha ,\varphi \left(
t\right) \right\rangle :=\left\langle \phi \left( \alpha \right)
,\left\langle {}_{0}D_{t}^{\alpha }h\left( t\right) ,\varphi \left( t\right)
\right\rangle \right\rangle ,\;\;\varphi \in \mathcal{D}\left(
\mathbb{R}
\right) .
\end{equation*}%
For details see a work by \citeasnoun{AOP}. Recall, $\mathcal{D}_{+}^{\prime
}\left(
\mathbb{R}
\right) $ denotes the space of distributions supported by $\left[ 0,\infty
\right) $ and $\left\langle h\left( t\right) ,\varphi \left( t\right)
\right\rangle $ denotes the action of a distribution $h\in \mathcal{D}%
_{+}^{\prime }\left(
\mathbb{R}
\right) $ on a test function $\varphi \in \mathcal{D}\left(
\mathbb{R}
\right) $ (see a book by \cite{vlad}).

In (\ref{3-1}), $\phi _{1}$ and $\phi _{2}$ denote constitutive functions or
distributions that are determined experimentally. The constitutive equations
of type (\ref{3-1}) were used earlier in papers by \citeasnoun{H-L}, %
\citeasnoun{a-2002-a}, \citeasnoun{APZ-1} and \citeasnoun{APZ-2}. There are
number of forms that $\phi _{1}$ and $\phi _{2}$ can take (see a paper by
\cite{H-L}, for example). In the sequel we assume that%
\begin{equation}
\phi _{1}\left( \alpha \right) :=a^{\alpha },\;\;\phi _{2}\left( \alpha
\right) :=b^{\alpha },\;\;\alpha \in \left( 0,1\right) ,\;a\leq b.
\label{a-b}
\end{equation}%
The restriction $a\leq b$ follows from the Second Law of Thermodynamics (see
for example papers by \cite{a-2002} and \cite{a-2003}). If $a=b,$ then (\ref%
{3-1}) reduces to the Hooke Law. The choice of $\phi _{1}$ and $\phi _{2}$
in the form (\ref{a-b}) is the simplest choice guaranteeing dimensional
homogeneity. Note that with $\phi _{1}\left( \mu \right) :=\delta \left( \mu
\right) +\tau _{\epsilon }^{\alpha }\delta \left( \mu -\alpha \right) $ and $%
\phi _{2}\left( \mu \right) :=E_{\infty }\tau _{\epsilon }^{\beta }\delta
\left( \mu -\beta \right) $ ($\delta $ denotes the Dirac distribution) we
obtain
\begin{equation}
\sigma +\tau _{\epsilon }^{\alpha }{}_{0}D_{t}^{\alpha }\sigma =E_{\infty
}\tau _{\epsilon }^{\beta }{}_{0}D_{t}^{\beta }\mathcal{\epsilon },
\label{3-2}
\end{equation}%
while with $\phi _{1}\left( \mu \right) :=\delta \left( \mu \right) +\tau
_{\varepsilon }^{\alpha }\delta \left( \mu -\alpha \right) $ and $\phi
_{2}\left( \mu \right) :=E_{0}\big(\delta \left( \mu \right) +\tau _{\sigma
}^{\alpha }\delta \left( \mu -\alpha \right) +\tau _{\sigma }^{\beta }\delta
\left( \mu -\beta \right) \big)$ we obtain%
\begin{equation}
\sigma +\tau _{\varepsilon }^{\alpha }{}_{0}D_{t}^{\alpha }\sigma
=E_{0}\left( 1+\tau _{\sigma }^{\alpha }{}_{0}D_{t}^{\alpha }+\tau _{\sigma
}^{\beta }{}_{0}D_{t}^{\beta }\right) \mathcal{\varepsilon }.  \label{3-3}
\end{equation}%
Recall, system (\ref{1}), (\ref{3}) and (\ref{3-2}), respectively system (%
\ref{1}), (\ref{3}) and (\ref{3-3}), was treated in a work by %
\citeasnoun{R-S}, respectively in a work by \citeasnoun{R-S1}. Also note
that the distributed order dissipation of type (\ref{3-1}) was also used in
the context of one degree of freedom mechanical systems in papers by %
\citeasnoun{a} and \citeasnoun{a1}.

Our aim is to find functions $u$ and $\sigma ,$ locally integrable on $%
\mathbb{R}
$ and equal to zero for $t<0,$ so that these functions satisfy (\ref{1}), (%
\ref{3}), (\ref{3-1}), for $x\in \left[ 0,L\right] $ and $t>0,$ as well as
the appropriate initial and boundary conditions. Actually, we will introduce
dimensionless quantities and transform the system (\ref{1}), (\ref{3}), (\ref%
{3-1}), subject to (\ref{IC}) and (\ref{BC-u}), into the system (\ref{sys-bd}%
), subject to (\ref{IC-bd}) and (\ref{BC-u-bd}).

The paper is organized as follows. In \S \ref{cfs} we introduce
dimensionless quantities, proceed by formal calculation and by the use of
the Laplace transformation we obtain solutions to (\ref{1}), (\ref{3}), (\ref%
{3-1}) in the convolution form. We impose initial conditions as well as
boundary conditions to (\ref{1}), (\ref{3}), (\ref{3-1}). Boundary
conditions describe a rod that is fixed at one of its ends, while the other
end is subject to a prescribed displacement $\Upsilon $ (this is the case of
stress relaxation if $\Upsilon =\Upsilon _{0}H,$ with $H$ being the
Heaviside function). Section \ref{P} is devoted to the calculation of the
inverse Laplace transformation, which leads to the explicit form of a
solution. More precisely, we investigate some properties of functions in
order to be able to apply the Cauchy residues theorem, which is used to
calculate the inverse Laplace transformation. We obtain displacement $u$ and
stress $\sigma $ for the boundary condition $\Upsilon =\Upsilon _{0}H$ in \S %
\ref{H}, as well as for $\Upsilon =\Upsilon _{0}H+F,$ where $F$ is an
appropriate function supported by $\left[ 0,\infty \right) ,$ in \S \ref%
{hplusnesto}. We conclude that solutions are locally integrable functions
supported by $\left[ 0,\infty \right) .$ Moreover, they are smooth functions
for $t>0.$ Numerical examples corresponding to stress relaxation are
presented in \S \ref{ne}. Concluding remarks are given in \S \ref{conc}.

\section{Convolution form of solutions. Formal calculation.\label{cfs}}

We prescribe initial conditions for system (\ref{1}), (\ref{3}), (\ref{3-1})
\begin{equation}
u\left( x,0\right) =0,\;\;\;\frac{\partial }{\partial t}u\left( x,0\right)
=0,\;\;\;\sigma \left( x,0\right) =0,\;\;\;\mathcal{E}\left( x,0\right)
=0,\;\;x\in \left[ 0,L\right] .  \label{IC}
\end{equation}%
We subject system (\ref{1}), (\ref{3}), (\ref{3-1}) to boundary conditions
corresponding to case of stress relaxation
\begin{equation}
u\left( 0,t\right) =0,\;\;\;u\left( L,t\right) =\Upsilon \left( t\right)
,\;\;t\in
\mathbb{R}
.  \label{BC-u}
\end{equation}%
Function $\Upsilon $ is locally integrable function equal to zero for $t<0.$

Introducing dimensionless quantities%
\begin{equation*}
\bar{x}=\frac{x}{L},\;\bar{t}=\frac{t}{L\sqrt{\frac{\rho }{E}}},\;\bar{u}=%
\frac{u}{L},\;\bar{\sigma}=\frac{\sigma }{E},\;\bar{\Upsilon}=\frac{\Upsilon
}{L},\;\bar{\phi}_{1}=\frac{\phi _{1}}{\left( L\sqrt{\frac{\rho }{E}}\right)
^{\alpha }},\;\bar{\phi}_{2}=\frac{\phi _{2}}{\left( L\sqrt{\frac{\rho }{E}}%
\right) ^{\alpha }},
\end{equation*}%
and using the fact that the fractional derivative transforms as%
\begin{equation*}
{}_{0}D_{\bar{t}}^{\alpha }u(\bar{t})=\left( L\sqrt{\frac{\rho }{E}}\right)
^{\alpha }{}_{0}D_{t}^{\alpha }u(t),
\end{equation*}%
we obtain, after omitting bar over dimensionless quantities, the following
system%
\begin{eqnarray}
&&\frac{\partial }{\partial x}\sigma \left( x,t\right) =\frac{\partial ^{2}}{%
\partial t^{2}}u\left( x,t\right) ,  \notag \\
&&\int_{0}^{1}\phi _{1}\left( \alpha \right) {}_{0}D_{t}^{\alpha }\sigma
\left( x,t\right) \mathrm{d}\alpha =\int_{0}^{1}\phi _{2}\left( \alpha
\right) {}_{0}D_{t}^{\alpha }\mathcal{E}\left( x,t\right) \mathrm{d}\alpha ,
\label{sys-bd} \\
&&\mathcal{E}\left( x,t\right) =\frac{\partial }{\partial x}u\left(
x,t\right) ,\;\;x\in \left[ 0,1\right] ,\;t>0.  \notag
\end{eqnarray}%
System (\ref{sys-bd}) is subject to initial
\begin{equation}
u\left( x,0\right) =0,\;\;\;\frac{\partial }{\partial t}u\left( x,0\right)
=0,\;\;\;\sigma \left( x,0\right) =0,\;\;\;\mathcal{E}\left( x,0\right)
=0,\;\;x\in \left[ 0,1\right] ,  \label{IC-bd}
\end{equation}%
and boundary conditions%
\begin{equation}
u\left( 0,t\right) =0,\;\;\;u\left( 1,t\right) =\Upsilon \left( t\right)
,\;\;t\in
\mathbb{R}
.  \label{BC-u-bd}
\end{equation}%
Additional assumptions on $\Upsilon $ will be given in \S \ref{H} and \S \ref%
{hplusnesto}.

In the sequel we assume that $\phi _{1}$ and $\phi _{2}$ are functions given
by (\ref{a-b}). Using (\ref{a-b}) in (\ref{sys-bd})$_{2}$ and formally
applying the Laplace transformation to (\ref{sys-bd}) and (\ref{IC-bd}), we
obtain%
\begin{eqnarray}
&&\frac{\partial }{\partial x}\tilde{\sigma}\left( x,s\right) =s^{2}\tilde{u}%
\left( x,s\right) ,  \notag \\
&&\tilde{\sigma}\left( x,s\right) \int_{0}^{1}\left( as\right) ^{\alpha }%
\mathrm{d}\alpha =\mathcal{\tilde{E}}\left( x,s\right) \int_{0}^{1}\left(
bs\right) ^{\alpha }\mathrm{d}\alpha ,  \label{S-LT} \\
&&\mathcal{\tilde{E}}\left( x,s\right) =\frac{\partial }{\partial x}\tilde{u}%
\left( x,s\right) ,\;\;x\in \left[ 0,1\right] ,\;s\in D.  \notag
\end{eqnarray}%
Recall, the Laplace transformation of $f\in L_{loc}^{1}\left( \mathbf{%
\mathbb{R}
}\right) ,$ $f\equiv 0$ in $\left( -\infty ,0\right] $ and $\left\vert
f\left( t\right) \right\vert \leq c\mathrm{e}^{at},$ $t>0,$ for some $a>0,$
is defined by%
\begin{equation*}
\tilde{f}\left( s\right) =\mathcal{L}\left[ f\left( t\right) \right] \left(
s\right) :=\int_{0}^{\infty }f\left( t\right) e^{-st}\mathrm{d}t,\;\;\func{Re%
}s>a
\end{equation*}%
and analytically continued into the appropriate domain $D.$ Domain $D$ for (%
\ref{S-LT}) is determined after (\ref{u-tilda}), bellow.

System (\ref{S-LT}) reduces to%
\begin{equation*}
\frac{\partial ^{2}}{\partial x^{2}}\tilde{u}\left( x,s\right) -s^{2}\frac{%
\ln \left( bs\right) }{\ln \left( as\right) }\frac{as-1}{bs-1}\tilde{u}%
\left( x,s\right) =0,\;\;x\in \left[ 0,1\right] ,\;s\in D,
\end{equation*}%
whose formal solution is%
\begin{equation}
\tilde{u}\left( x,s\right) =C_{1}\left( s\right) \mathrm{e}^{xs\sqrt{\frac{%
\ln \left( bs\right) }{\ln \left( as\right) }\frac{as-1}{bs-1}}}+C_{2}\left(
s\right) \mathrm{e}^{-xs\sqrt{\frac{\ln \left( bs\right) }{\ln \left(
as\right) }\frac{as-1}{bs-1}}},\;\;x\in \left[ 0,1\right] ,\;s\in D,
\label{u-tilda}
\end{equation}%
where $C_{1}$ and $C_{2}$ are arbitrary functions which will be determined
from the boundary conditions. Since the natural logarithm has the branch
point at $s=0,$ we have $D=%
\mathbb{C}
\backslash \left( -\infty ,0\right] $ and this will be used in proposition %
\ref{propP}.

Applying (\ref{BC-u-bd})$_{1}$ we obtain $C=C_{1}=-C_{2}$ and thus%
\begin{equation*}
\tilde{u}\left( x,s\right) =C\left( s\right) \left( \mathrm{e}^{xs\sqrt{%
\frac{\ln \left( bs\right) }{\ln \left( as\right) }\frac{as-1}{bs-1}}}-%
\mathrm{e}^{-xs\sqrt{\frac{\ln \left( bs\right) }{\ln \left( as\right) }%
\frac{as-1}{bs-1}}}\right) ,\;\;x\in \left[ 0,1\right] ,\;s\in
\mathbb{C}
\backslash \left( -\infty ,0\right] .
\end{equation*}%
Using (\ref{BC-u-bd})$_{2}$ in the previous expression, it follows%
\begin{equation}
\tilde{u}\left( x,s\right) =\tilde{\Upsilon}\left( s\right) \tilde{P}\left(
x,s\right) ,\;\;x\in \left[ 0,1\right] ,\;s\in
\mathbb{C}
\backslash \left( -\infty ,0\right] .  \label{4-0}
\end{equation}%
We introduced $\tilde{P}$ as%
\begin{equation}
\tilde{P}\left( x,s\right) =\frac{\sinh \left( xs\sqrt{\frac{\ln \left(
bs\right) }{\ln \left( as\right) }\frac{as-1}{bs-1}}\right) }{\sinh \left( s%
\sqrt{\frac{\ln \left( bs\right) }{\ln \left( as\right) }\frac{as-1}{bs-1}}%
\right) },\;\;x\in \left[ 0,1\right] ,\;s\in
\mathbb{C}
\backslash \left( -\infty ,0\right] .  \label{4-1}
\end{equation}%
Note that $P\left( x,\cdot \right) $ is a distribution supported by $\left[
0,\infty \right) .$ It is clear that for $x=1,$ $P\left( 1,t\right) =\delta
\left( t\right) ,$ $t\in
\mathbb{R}
.$ Actually, we will calculate $P\ $and show that for every $x\in \left[ 0,1%
\right] ,$ $P\left( x,\cdot \right) $ is a locally integrable function on $%
\mathbb{R}
,$ equal to zero for $t<0.$

Since $\Upsilon $ and $P$ are supported by $\left[ 0,\infty \right) ,$
displacement $u$ is given by%
\begin{eqnarray}
u\left( x,t\right) &=&\Upsilon \left( t\right) \ast P\left( x,t\right)
,\;\;x\in \left[ 0,1\right] ,\;t\in
\mathbb{R}
,\;\;\text{so that}  \label{4-2} \\
u\left( x,t\right) &=&0,\;\;x\in \left[ 0,1\right] ,\;t<0,  \notag
\end{eqnarray}%
where we use $\ast $ to denote the convolution. Recall if $f,g\in
L_{loc}^{1}\left( \mathbf{%
\mathbb{R}
}\right) ,$ $\limfunc{supp}f,g\subset \left[ 0,\infty \right) ,$ then $%
\left( f\ast g\right) \left( t\right) :=\int_{0}^{t}f\left( \tau \right)
g\left( t-\tau \right) \mathrm{d}\tau ,$ $t\in
\mathbb{R}
.$ Calculation of (\ref{4-2}) will be done by the use of the Laplace
inversion formula applied to (\ref{4-0}). Moreover, we will show that $%
u_{H}\left( x,t\right) =H\left( t\right) \ast P\left( x,t\right) ,$ $t\in
\mathbb{R}
,$ is a continuous function, equal to zero for $t<0.$

Next, we use (\ref{S-LT})$_{2}$ and (\ref{S-LT})$_{3}$ and obtain
\begin{equation}
\tilde{\sigma}\left( x,s\right) =\frac{\ln \left( as\right) }{\ln \left(
bs\right) }\frac{bs-1}{as-1}\frac{\partial }{\partial x}\tilde{u}\left(
x,s\right) ,\;\;x\in \left[ 0,1\right] ,\;s\in
\mathbb{C}
\backslash \left( -\infty ,0\right] .  \label{4-3-1}
\end{equation}%
In order to determine $\sigma ,$ we use (\ref{4-0}) and (\ref{4-1}) in (\ref%
{4-3-1}), so that
\begin{equation}
\tilde{\sigma}\left( x,s\right) =s\tilde{\Upsilon}\left( s\right) \tilde{T}%
\left( x,s\right) ,\;\;x\in \left[ 0,1\right] ,\;s\in
\mathbb{C}
\backslash \left( -\infty ,0\right] ,  \label{4-3}
\end{equation}%
where
\begin{equation}
\tilde{T}\left( x,s\right) =\sqrt{\frac{\ln \left( as\right) }{\ln \left(
bs\right) }\frac{bs-1}{as-1}}\frac{\cosh \left( xs\sqrt{\frac{\ln \left(
bs\right) }{\ln \left( as\right) }\frac{as-1}{bs-1}}\right) }{\sinh \left( s%
\sqrt{\frac{\ln \left( bs\right) }{\ln \left( as\right) }\frac{as-1}{bs-1}}%
\right) },\;\;x\in \left[ 0,1\right] ,\;s\in
\mathbb{C}
\backslash \left( -\infty ,0\right] .  \label{7-1}
\end{equation}%
Note that $T$ is a locally integrable function for $t>0.$ For the
determination of $\sigma ,$ we will use the Laplace inversion formula
applied to (\ref{4-3}) and obtain
\begin{equation}
\sigma \left( x,t\right) =\frac{\mathrm{d}}{\mathrm{d}t}%
\big(%
\Upsilon \left( t\right) \ast T\left( x,t\right)
\big)%
,\;\;x\in \left[ 0,1\right] ,\;t\in
\mathbb{R}
,  \label{11}
\end{equation}%
where the derivative is understood in the sense of distributions. Again, $%
\sigma \left( x,t\right) =0$ for $x\in \left[ 0,1\right] ,\;t<0.$ For the
detailed account see \S \textit{\ref{dts}}.

\section{Explicit forms of solutions\label{P}}

In this section we will calculate inverse Laplace transformations of
distributions and functions on $%
\mathbb{R}
$ supported by $\left[ 0,\infty \right) .$ For that purposes, let us define%
\begin{equation}
M\left( s\right) :=\sqrt{\frac{\ln \left( bs\right) }{\ln \left( as\right) }%
\frac{as-1}{bs-1}},\;\;s\in
\mathbb{C}
\backslash \left( -\infty ,0\right] .  \label{M}
\end{equation}%
In the sequel we will write $A\left( x\right) \sim B\left( x\right) $ if $%
\lim\limits_{x\rightarrow \infty }\frac{A\left( x\right) }{B\left( x\right) }%
=1.$ Next proposition establishes some properties of $M.$

\begin{proposition}
\label{propP}

\begin{enumerate}
\item[(i)] $M$ is an analytic function in $s\in
\mathbb{C}
\backslash \left( -\infty ,0\right] ;$

\item[(ii)] $\lim\limits_{\substack{ s\rightarrow 0  \\ s\in
\mathbb{C}
\backslash \left( -\infty ,0\right] }}M\left( s\right) =1$\ and$%
\;\lim\limits _{\substack{ \left\vert s\right\vert \rightarrow \infty  \\ %
s\in
\mathbb{C}
\backslash \left( -\infty ,0\right] }}M\left( s\right) =\sqrt{\frac{a}{b}}.$

\item[(iii)] Let $p\in \left( 0,s_{0}\right) ,$ $s_{0}>0.$ Then
\begin{equation*}
M\left( p\pm \mathrm{i}R\right) \sim \sqrt{\frac{a}{b}}\frac{1}{\ln \left(
aR\right) }\sqrt[4]{\left( \ln \left( aR\right) \ln \left( bR\right) \right)
^{2}+\left( \frac{\pi }{2}\ln \frac{b}{a}\right) ^{2}}\mathrm{e}^{\mp
\mathrm{i}\arctan \frac{\frac{\pi }{2}\ln \frac{b}{a}}{\ln \left( aR\right)
\ln \left( bR\right) }}
\end{equation*}%
as $R\rightarrow \infty .$
\end{enumerate}
\end{proposition}

\begin{proof}
We first prove (i). The only points where $M$ could be singular are $s=\frac{%
1}{a}$ and $s=\frac{1}{b}.$ Since,
\begin{equation*}
\frac{\ln \left( bs\right) }{bs-1}=\frac{\ln \left( 1+\left( bs-1\right)
\right) }{bs-1}=\sum_{n=0}^{\infty }\frac{\left( -1\right) ^{n}\left(
bs-1\right) ^{n}}{n+1},\;\;\left\vert bs-1\right\vert \in \left( -1,1\right]
,
\end{equation*}%
it is obvious that $s=\frac{1}{b}$ is a regular point of $M.$ Similar
arguments hold for $s=\frac{1}{a}.$ Limits in (ii) can easily be calculated.
In order to prove (iii), let us introduce $\mu =\sqrt{p^{2}+R^{2}}$ and $\nu
=\arctan \frac{R}{p}.$ It is obvious that $\mu \sim R,$ $\nu \sim \frac{\pi
}{2},$ as $R\rightarrow \infty .$ Then $M\left( p\pm \mathrm{i}R\right) $
becomes%
\begin{eqnarray*}
&&M\left( p\pm \mathrm{i}R\right) \\
&&\;\;\;\;=\sqrt{\frac{\ln \left( a\mu \right) \ln \left( b\mu \right) +\nu
^{2}\mp \mathrm{i}\nu \ln \frac{b}{a}}{\ln ^{2}\left( a\mu \right) +\nu ^{2}}%
}\sqrt{\frac{\left( ap-1\right) \left( bp-1\right) +abR^{2}\pm \mathrm{i}%
R\left( b-a\right) }{\left( bp-1\right) ^{2}+\left( bR\right) ^{2}}} \\
&&\;\;\;\;=\left( \frac{\left[ abR^{2}+\left( ap-1\right) \left( bp-1\right) %
\right] \left[ \ln \left( a\mu \right) \ln \left( b\mu \right) +\nu ^{2}%
\right] +R\nu \left( b-a\right) \ln \frac{b}{a}}{\left( \left( bR\right)
^{2}+\left( bp-1\right) ^{2}\right) \left( \ln ^{2}\left( a\mu \right) +\nu
^{2}\right) }\right. \\
&&\;\;\;\;\;\;\;\;\left. \pm \mathrm{i}\frac{R\left( b-a\right) \left[ \ln
\left( a\mu \right) \ln \left( b\mu \right) +\nu ^{2}\right] -\nu \ln \frac{b%
}{a}\left[ abR^{2}+\left( ap-1\right) \left( bp-1\right) \right] }{\left(
\left( bR\right) ^{2}+\left( bp-1\right) ^{2}\right) \left( \ln ^{2}\left(
a\mu \right) +\nu ^{2}\right) }\right) ^{\frac{1}{2}} \\
&&\;\;\;\;\sim \sqrt{\frac{a}{b}}\frac{1}{\ln \left( aR\right) }\sqrt{\ln
\left( aR\right) \ln \left( bR\right) \mp \mathrm{i}\frac{\pi }{2}\ln \frac{b%
}{a}},\;\;\text{as}\;\;R\rightarrow \infty .
\end{eqnarray*}
\end{proof}

\subsection{Determination of displacement $u$ in case of stress relaxation
\label{DSR}}

We investigate properties of $\tilde{P},$ given by (\ref{4-1}), and find a
solution to (\ref{sys-bd}), (\ref{IC-bd}), (\ref{BC-u-bd}), with weight
functions given by (\ref{a-b}), in two steps. First, we find a solution with
the boundary condition (\ref{BC-u-bd}) in the case when
\begin{equation}
\Upsilon \left( t\right) =\Upsilon _{0}H\left( t\right) ,\;\;\Upsilon
_{0}>0,\;t\in
\mathbb{R}
.  \label{5}
\end{equation}%
Then we proceed to a more general case, when $\Upsilon $ is assumed to be of
the form%
\begin{equation}
\Upsilon \left( t\right) =\Upsilon _{0}H\left( t\right) +F\left( t\right)
,\;\;t\in
\mathbb{R}
,  \label{6}
\end{equation}%
where $F$ is a locally integrable function, equal to zero on $\left( -\infty
,0\right] .$ Additional assumptions on $F$ will be postulated in \S \ref%
{hplusnesto}.

Let us examine properties of $\tilde{P},$ given by (\ref{4-1}). Clearly, it
has complex conjugated poles at ${}_{P}s_{n}^{\left( \pm \right) },$ $n\in
\mathbb{N}
.$ Poles are solutions of
\begin{equation}
\sinh \left( sM\left( s\right) \right) =0,\;\;\text{i.e.}\;\;sM\left(
s\right) =\pm n\mathrm{i}\pi .  \label{21-1}
\end{equation}%
Let us examine the position and the multiplicity of solutions to (\ref{21-1}%
).

\begin{proposition}
\label{levo}There are infinitely many solutions ${}_{P}s_{n}^{\left( \pm
\right) },$ $n\in
\mathbb{N}
,$ of (\ref{21-1}), such that%
\begin{eqnarray}
&&\func{Re}\left( {}_{P}s_{n}^{\left( \pm \right) }\right) \approx -\frac{%
\frac{\pi }{4}\ln \frac{b}{a}\sqrt{\frac{b}{a}}n\pi }{\ln \left( \sqrt{ab}%
n\pi \right) \ln \left( b\sqrt{\frac{b}{a}}n\pi \right) },  \label{Re s} \\
&&\func{Im}\left( {}_{P}s_{n}^{\left( \pm \right) }\right) \approx \pm
R\approx \pm \sqrt{\frac{b}{a}}n\pi ,  \label{Im s}
\end{eqnarray}%
as $n\rightarrow \infty .$ Moreover, there exists $n_{0}\in
\mathbb{N}
,$ such that poles ${}_{P}s_{n}^{\left( \pm \right) },$ for $n>n_{0},$ are
simple.
\end{proposition}

\begin{proof}
Let us square (\ref{21-1}) and put ${}_{P}s_{n}^{\left( \pm \right) }=R%
\mathrm{e}^{\mathrm{i}\phi },$ $\phi \in \left( -\pi ,\pi \right) .$ Then,
after separation of real and imaginary parts, we obtain
\begin{eqnarray}
&&R^{2}\cos \left( 2\phi \right) \func{Re}\left( M^{2}\left( R\mathrm{e}^{%
\mathrm{i}\phi }\right) \right) -R^{2}\sin \left( 2\phi \right) \func{Im}%
\left( M^{2}\left( R\mathrm{e}^{\mathrm{i}\phi }\right) \right) =-n^{2}\pi
^{2},  \label{1-a} \\
&&R^{2}\sin \left( 2\phi \right) \func{Re}\left( M^{2}\left( R\mathrm{e}^{%
\mathrm{i}\phi }\right) \right) +R^{2}\cos \left( 2\phi \right) \func{Im}%
\left( M^{2}\left( R\mathrm{e}^{\mathrm{i}\phi }\right) \right) =0.
\label{2-a}
\end{eqnarray}%
By the use of (\ref{M}), real and imaginary parts of $M^{2}\left( R\mathrm{e}%
^{\mathrm{i}\phi }\right) $ are
\begin{eqnarray*}
\func{Re}\left( M^{2}\left( R\mathrm{e}^{\mathrm{i}\phi }\right) \right) &=&%
\frac{\left( \ln \left( aR\right) \ln \left( bR\right) +\phi ^{2}\right)
\left( abR^{2}-\left( a+b\right) R\cos \phi +1\right) }{\left( \ln
^{2}\left( aR\right) +\phi ^{2}\right) \left( b^{2}R^{2}-2bR\cos \phi
+1\right) } \\
&&+\frac{\ln \frac{b}{a}\left( b-a\right) R\phi \sin \phi }{\left( \ln
^{2}\left( aR\right) +\phi ^{2}\right) \left( b^{2}R^{2}-2bR\cos \phi
+1\right) }, \\
\func{Im}\left( M^{2}\left( R\mathrm{e}^{\mathrm{i}\phi }\right) \right) &=&-%
\frac{\phi \ln \frac{b}{a}\left( abR^{2}-\left( a+b\right) R\cos \phi
+1\right) }{\left( \ln ^{2}\left( aR\right) +\phi ^{2}\right) \left(
b^{2}R^{2}-2bR\cos \phi +1\right) } \\
&&+\frac{+\left( b-a\right) R\sin \phi \left( \ln \left( aR\right) \ln
\left( bR\right) +\phi ^{2}\right) }{\left( \ln ^{2}\left( aR\right) +\phi
^{2}\right) \left( b^{2}R^{2}-2bR\cos \phi +1\right) }.
\end{eqnarray*}%
Letting $R\rightarrow \infty ,$ previous expressions are written as
\begin{eqnarray}
\func{Re}\left( M^{2}\left( R\mathrm{e}^{\mathrm{i}\phi }\right) \right)
&\approx &\frac{abR^{2}\ln \left( aR\right) \ln \left( bR\right) }{%
b^{2}R^{2}\ln ^{2}\left( aR\right) }=\frac{a}{b}\frac{\ln \left( bR\right) }{%
\ln \left( aR\right) },  \label{3-a} \\
\func{Im}\left( M^{2}\left( R\mathrm{e}^{\mathrm{i}\phi }\right) \right)
&\approx &-\frac{ab\ln \frac{b}{a}R^{2}\phi }{b^{2}R^{2}\ln ^{2}\left(
aR\right) }=-\frac{a}{b}\ln \frac{b}{a}\phi \frac{1}{\ln ^{2}\left(
aR\right) }.  \label{4-a}
\end{eqnarray}%
Using (\ref{2-a}), (\ref{3-a}) and (\ref{4-a}), we obtain%
\begin{equation}
\tan \left( 2\phi \right) =-\frac{\func{Im}\left( M^{2}\left( R\mathrm{e}^{%
\mathrm{i}\phi }\right) \right) }{\func{Re}\left( M^{2}\left( R\mathrm{e}^{%
\mathrm{i}\phi }\right) \right) }\approx \phi \frac{\ln \frac{b}{a}}{\ln
\left( aR\right) \ln \left( bR\right) }.  \label{tg fi/fi}
\end{equation}%
Let $\phi \in \left( 0,\pi \right) .$ Then $\frac{\tan \left( 2\phi \right)
}{\phi }>0$ and $\frac{\tan \left( 2\phi \right) }{\phi }\rightarrow 0$ as $%
R\rightarrow \infty .$ Hence, $\phi \in \left( 0,\frac{\pi }{4}\right) $ or $%
\phi \in \left( \frac{\pi }{2},\frac{3\pi }{4}\right) .$ Since $\phi \neq 0$
and $\tan \left( 2\phi \right) \rightarrow 0,$ it follows that $\phi
\rightarrow \frac{\pi }{2}$ from the interval $\phi \in \left( \frac{\pi }{2}%
,\frac{3\pi }{4}\right) .$ Therefore, by (\ref{tg fi/fi}), we have
\begin{equation}
\sin \phi \approx 1-\left( \frac{\frac{\pi }{2}\ln \frac{b}{a}}{2\ln \left(
aR\right) \ln \left( bR\right) }\right) ^{2}\approx 1,\;\;\cos \phi \approx -%
\frac{\frac{\pi }{2}\ln \frac{b}{a}}{2\ln \left( aR\right) \ln \left(
bR\right) }.  \label{kosinus}
\end{equation}%
Inserting (\ref{3-a}), (\ref{4-a}) and (\ref{kosinus}) in (\ref{1-a}), we
obtain%
\begin{equation}
\frac{1}{\sqrt{\ln ^{2}\left( aR\right) \ln ^{2}\left( bR\right) +\left(
\frac{\pi }{2}\ln \frac{b}{a}\right) ^{2}}}\left( \ln ^{2}\left( bR\right) +%
\frac{\left( \frac{\pi }{2}\ln \frac{b}{a}\right) ^{2}}{\ln ^{2}\left(
aR\right) }\right) \approx \frac{b}{a}\frac{n^{2}\pi ^{2}}{R^{2}}\approx 1.
\label{R}
\end{equation}%
Thus, real and imaginary parts of ${}_{P}s_{n}^{\left( \pm \right) },$ as $%
R\rightarrow \infty ,$ obtained by (\ref{kosinus}) and (\ref{R}), are as
stated in proposition.

In order to prove that solutions to (\ref{21-1}) are simple for $n>n_{0},$
we define
\begin{equation*}
f\left( s\right) :=\sinh \left( sM\left( s\right) \right) ,\;\;s\in
\mathbb{C}
\backslash \left( -\infty ,0\right] .
\end{equation*}%
Then ($s\in
\mathbb{C}
\backslash \left( -\infty ,0\right] $)%
\begin{equation*}
\frac{\mathrm{d}}{\mathrm{d}s}f\left( s\right) =M\left( s\right) \left( 1-%
\frac{\ln \frac{b}{a}}{2\ln \left( as\right) \ln \left( bs\right) }+\frac{%
\left( b-a\right) s}{2\left( as-1\right) \left( bs-1\right) }\right) \cosh
\left( sM\left( s\right) \right) .
\end{equation*}%
Solutions to $f\left( s\right) =0$ are given by (\ref{21-1}), and so, as $%
\left\vert {}_{P}s_{n}^{\left( \pm \right) }\right\vert \rightarrow \infty ,$%
\begin{equation*}
\left. \frac{\mathrm{d}}{\mathrm{d}s}f\left( s\right) \right\vert
_{s={}_{P}s_{n}^{\left( \pm \right) }}\sim \left( -1\right) ^{n}\left[
M\left( s\right) \left( 1-\frac{\ln \frac{b}{a}}{2\ln \left( as\right) \ln
\left( bs\right) }+\frac{\left( b-a\right) s}{2\left( as-1\right) \left(
bs-1\right) }\right) \right] _{s={}_{P}s_{n}^{\left( \pm \right) }}.
\end{equation*}%
By proposition \ref{propP}, $M\sim \sqrt{\frac{a}{b}}$ and this implies that
\begin{equation*}
\left. \frac{\mathrm{d}}{\mathrm{d}s}f\left( s\right) \right\vert
_{s={}_{P}s_{n}^{\left( \pm \right) }}\sim \left( -1\right) ^{n}\sqrt{\frac{a%
}{b}}\;\;\text{as}\;\;\left\vert {}_{P}s_{n}^{\left( \pm \right)
}\right\vert \rightarrow \infty .
\end{equation*}%
Thus, for large $\left\vert {}_{P}s_{n}^{\left( \pm \right) }\right\vert $
we have $\left. \frac{\mathrm{d}}{\mathrm{d}s}f\left( s\right) \right\vert
_{s={}_{P}s_{n}^{\left( \pm \right) }}\neq 0$ and solutions are simple for $%
n>n_{0}$.
\end{proof}

\subsubsection{Case $\Upsilon =\Upsilon _{0}H$\label{H}}

This is the case that has physical importance, since we obtain displacement
in case of stress relaxation test. Formally, we write (\ref{4-2}) with (\ref%
{5}) as%
\begin{equation}
u_{H}\left( x,t\right) =\Upsilon _{0}H\left( t\right) \ast P\left(
x,t\right) ,\;\;x\in \left[ 0,1\right] ,\;t\in
\mathbb{R}
.  \label{u-h}
\end{equation}

The following theorem is on existence and properties of $u_{H}.$

\begin{theorem}
Let $\Upsilon =\Upsilon _{0}H$ and let $\phi _{1}$ and $\phi _{2}$ be given
by (\ref{a-b}). Then the solution to (\ref{sys-bd}), (\ref{IC-bd}), (\ref%
{BC-u-bd}) is given by (\ref{u-h}), where%
\begin{eqnarray}
P\left( x,t\right) &=&\frac{1}{2\pi \mathrm{i}}\dint\nolimits_{0}^{\infty
}\left( \frac{\sinh \left( xqM\left( q\mathrm{e}^{-\mathrm{i}\pi }\right)
\right) }{\sinh \left( qM\left( q\mathrm{e}^{-\mathrm{i}\pi }\right) \right)
}-\frac{\sinh \left( xqM\left( q\mathrm{e}^{\mathrm{i}\pi }\right) \right) }{%
\sinh \left( qM\left( q\mathrm{e}^{\mathrm{i}\pi }\right) \right) }\right)
\mathrm{e}^{-qt}\mathrm{d}q  \notag \\
&&+\sum_{n=1}^{\infty }\left[ \func{Res}\left( \tilde{P}\left( x,s\right)
\mathrm{e}^{st},{}_{P}s_{n}^{\left( +\right) }\right) \right.  \notag \\
&&\left. +\func{Res}\left( \tilde{P}\left( x,s\right) \mathrm{e}%
^{st},{}_{P}s_{n}^{\left( -\right) }\right) \right] ,\;\;x%
\begin{tabular}{l}
$\in $%
\end{tabular}%
\left[ 0,1\right] ,\;t%
\begin{tabular}{l}
\TEXTsymbol{>}%
\end{tabular}%
0,  \label{P1} \\
P\left( x,t\right) &=&0,\;\;x\in \left[ 0,1\right] ,\;t<0.  \label{uhtm0}
\end{eqnarray}%
The residues are given by%
\begin{equation}
\func{Res}\left( \tilde{P}\left( x,s\right) \mathrm{e}^{st},{}_{P}s_{n}^{%
\left( \pm \right) }\right) =\left[ \frac{\sinh \left( xsM\left( s\right)
\right) }{\frac{\mathrm{d}}{\mathrm{d}s}\left[ \sinh \left( sM\left(
s\right) \right) \right] }\mathrm{e}^{st}\right] _{s={}_{P}s_{n}^{\left( \pm
\right) }}  \label{res-P}
\end{equation}%
and simple poles ${}_{P}s_{n}^{\left( \pm \right) },$ for $n>n_{0},$ are
solutions of (\ref{21-1}). Function $P$ is real-valued, locally integrable
on $%
\mathbb{R}
$ and smooth for $t>0.$

The explicit form of solution is%
\begin{eqnarray}
u_{H}\left( x,t\right) &=&\frac{\Upsilon _{0}}{2\pi \mathrm{i}}%
\dint\nolimits_{0}^{\infty }\left( \frac{\sinh \left( xqM\left( q\mathrm{e}%
^{-\mathrm{i}\pi }\right) \right) }{\sinh \left( qM\left( q\mathrm{e}^{-%
\mathrm{i}\pi }\right) \right) }-\frac{\sinh \left( xqM\left( q\mathrm{e}^{%
\mathrm{i}\pi }\right) \right) }{\sinh \left( qM\left( q\mathrm{e}^{\mathrm{i%
}\pi }\right) \right) }\right) \frac{1-\mathrm{e}^{-qt}}{q}\mathrm{d}q
\notag \\
&&+\int\nolimits_{0}^{t}%
\Bigg(%
\sum_{n=1}^{\infty }\left[ \func{Res}\left( \tilde{P}\left( x,s\right)
\mathrm{e}^{s\tau },{}_{P}s_{n}^{\left( +\right) }\right) \right.  \notag \\
&&\left. +\func{Res}\left( \tilde{P}\left( x,s\right) \mathrm{e}^{s\tau
},{}_{P}s_{n}^{\left( -\right) }\right) \right]
\Bigg)%
\mathrm{d}\tau ,\;\;x%
\begin{tabular}{l}
$\in $%
\end{tabular}%
\left[ 0,1\right] ,\;t%
\begin{tabular}{l}
\TEXTsymbol{>}%
\end{tabular}%
0,  \label{uha} \\
u_{H}\left( x,t\right) &=&0,\;\;x\in \left[ 0,1\right] ,\;t<0.  \label{uha0}
\end{eqnarray}%
Function $u_{H}$ is continuous at $t=0.$
\end{theorem}

\begin{proof}
We calculate $P\left( x,t\right) ,$ $x\in \left[ 0,1\right] ,$ $t\in
\mathbb{R}
,$ by the integration over a suitable contour.

Let $t>0.$ The Cauchy residues theorem yields
\begin{equation}
\oint\nolimits_{\Gamma }\tilde{P}\left( x,s\right) \mathrm{e}^{st}\mathrm{d}%
s=2\pi \mathrm{i}\sum_{n=1}^{\infty }\left[ \func{Res}\left( \tilde{P}\left(
x,s\right) \mathrm{e}^{st},{}_{P}s_{n}^{\left( +\right) }\right) +\func{Res}%
\left( \tilde{P}\left( x,s\right) \mathrm{e}^{st},{}_{P}s_{n}^{\left(
-\right) }\right) \right] ,  \label{KF-P}
\end{equation}%
where $\Gamma =\Gamma _{1}\cup \Gamma _{2}\cup \Gamma _{3}\cup \Gamma
_{\varepsilon }\cup \Gamma _{4}\cup \Gamma _{5}\cup \Gamma _{6}\cup \gamma
_{0},$ so that all poles lie inside the contour $\Gamma $ (see figure \ref%
{fig-1}).
\begin{figure}[h]
\centering
\includegraphics[scale=0.45]{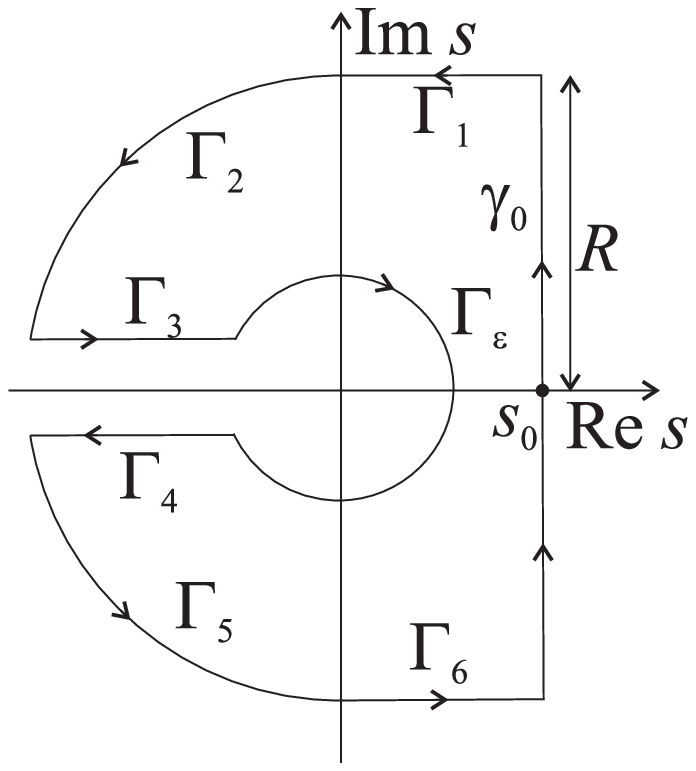}
\caption{Integration contour $\Gamma $}
\label{fig-1}
\end{figure}

First we show that the series of residues in (\ref{P1}) is convergent. By
proposition \ref{levo}, poles ${}_{P}s_{n}^{\left( \pm \right) }$ of $\tilde{%
P},$ given by (\ref{4-1}), are simple for $n>n_{0}$. Then residues in (\ref%
{KF-P}) can be calculated as it is given in (\ref{res-P}). We use (\ref{21-1}%
) to write (\ref{res-P}) as%
\begin{eqnarray}
&&\func{Res}\left( \tilde{P}\left( x,s\right) \mathrm{e}^{st},{}_{P}s_{n}^{%
\left( \pm \right) }\right)   \notag \\
&&\;\;\;\;=\left( -1\right) ^{n}\frac{\sin \left( n\pi x\right) }{n\pi }%
\left[ \frac{s\mathrm{e}^{st}}{1-\frac{\ln \frac{b}{a}}{2\ln \left(
as\right) \ln \left( bs\right) }+\frac{s\left( b-a\right) }{2\left(
as-1\right) \left( bs-1\right) }}\right] _{s={}_{P}s_{n}^{\left( \pm \right)
}},\;\;n>n_{0}.  \notag \\
&&  \label{rez}
\end{eqnarray}%
Let ${}_{P}s_{n}^{\left( \pm \right) }=R\mathrm{e}^{\pm \mathrm{i}\phi },$
then (\ref{rez}) transforms into
\begin{eqnarray*}
&&\func{Res}\left( \tilde{P}\left( x,s\right) \mathrm{e}^{st},{}_{P}s_{n}^{%
\left( \pm \right) }\right)  \\
&&\;\;\;\;=\left( -1\right) ^{n}\frac{\sin \left( n\pi x\right) }{n\pi }%
\frac{R\mathrm{e}^{Rt\cos \phi }\mathrm{e}^{\pm \mathrm{i}\left( \phi
+Rt\sin \phi \right) }}{\left[ 1-\frac{\ln \frac{b}{a}}{2\ln \left(
as\right) \ln \left( bs\right) }+\frac{s\left( b-a\right) }{2\left(
as-1\right) \left( bs-1\right) }\right] _{s=R\mathrm{e}^{\pm \mathrm{i}\phi
}}} \\
&&\;\;\;\;=\left( -1\right) ^{n}\frac{\sin \left( n\pi x\right) }{n\pi }%
\frac{R\mathrm{e}^{Rt\cos \phi }\left[ \cos \left( \phi +Rt\sin \phi \right)
\pm \mathrm{i\sin }\left( \phi +Rt\sin \phi \right) \right] }{\left[ 1-\frac{%
\ln \frac{b}{a}}{2\ln \left( as\right) \ln \left( bs\right) }+\frac{s\left(
b-a\right) }{2\left( as-1\right) \left( bs-1\right) }\right] _{s=R\mathrm{e}%
^{\pm \mathrm{i}\phi }}},
\end{eqnarray*}%
and therefore, for $n>n_{0},$ we have%
\begin{eqnarray}
&&\func{Res}\left( \tilde{P}\left( x,s\right) \mathrm{e}^{st},{}_{P}s_{n}^{%
\left( +\right) }\right) +\func{Res}\left( \tilde{P}\left( x,s\right)
\mathrm{e}^{st},{}_{P}s_{n}^{\left( -\right) }\right)
\begin{tabular}{l}
=%
\end{tabular}
\notag \\
&&\;\;\;\;\left( -1\right) ^{n}\frac{\sin \left( n\pi x\right) }{n\pi }R%
\mathrm{e}^{Rt\cos \phi }  \notag \\
&&\;\;\;\;\times \left( \frac{\cos \left( \phi +Rt\sin \phi \right) +\mathrm{%
i\sin }\left( \phi +Rt\sin \phi \right) }{\left[ 1-\frac{\ln \frac{b}{a}}{%
2\ln \left( as\right) \ln \left( bs\right) }+\frac{s\left( b-a\right) }{%
2\left( as-1\right) \left( bs-1\right) }\right] _{s=R\mathrm{e}^{\mathrm{i}%
\phi }}}\right.   \notag \\
&&\;\;\;\;\left. +\frac{\cos \left( \phi +Rt\sin \phi \right) -\mathrm{i\sin
}\left( \phi +Rt\sin \phi \right) }{\left[ 1-\frac{\ln \frac{b}{a}}{2\ln
\left( as\right) \ln \left( bs\right) }+\frac{s\left( b-a\right) }{2\left(
as-1\right) \left( bs-1\right) }\right] _{s=R\mathrm{e}^{-\mathrm{i}\phi }}}%
\right) .  \label{oc}
\end{eqnarray}%
Let $n\rightarrow \infty $ (then also $\left\vert {}_{P}s_{n}^{\left( \pm
\right) }\right\vert \rightarrow \infty ,$ i.e. $R\rightarrow \infty $). Then%
\begin{equation*}
\left\vert \left[ 1-\frac{\ln \frac{b}{a}}{2\ln \left( as\right) \ln \left(
bs\right) }+\frac{s\left( b-a\right) }{2\left( as-1\right) \left(
bs-1\right) }\right] _{s=R\mathrm{e}^{\pm \mathrm{i}\phi }}\right\vert
\rightarrow 1,
\end{equation*}%
and (\ref{oc}), as $n\rightarrow \infty ,$ becomes
\begin{eqnarray*}
&&\left\vert \func{Res}\left( \tilde{P}\left( x,s\right) \mathrm{e}%
^{st},{}_{P}s_{n}^{\left( +\right) }\right) +\func{Res}\left( \tilde{P}%
\left( x,s\right) \mathrm{e}^{st},{}_{P}s_{n}^{\left( -\right) }\right)
\right\vert  \\
&&\;\;\;\;\approx 2\left\vert \frac{\sin \left( n\pi x\right) }{n\pi }R%
\mathrm{e}^{Rt\cos \phi }\cos \left( \phi +Rt\sin \phi \right) \right\vert .
\end{eqnarray*}%
Proposition \ref{levo}, (\ref{Re s}), implies that%
\begin{equation*}
\func{Re}\left( {}_{P}s_{n}^{\left( \pm \right) }\right) \approx -\frac{\pi
}{4}\ln \frac{b}{a}\sqrt{\frac{b}{a}}\pi \frac{n}{\ln \left( \sqrt{ab}n\pi
\right) \ln \left( b\sqrt{\frac{b}{a}}n\pi \right) }\leq -C\sqrt{n}%
,\;\;n>n_{0}.
\end{equation*}%
Also, by (\ref{Im s}), we have that $\frac{R}{n}\approx \sqrt{\frac{b}{a}}%
\pi $. This implies that summands in (\ref{KF-P}) can be estimated by $K%
\mathrm{e}^{-Ct\sqrt{n}},$ which implies the convergence of the sum of
residues in (\ref{KF-P}).

Second, we calculate the integral over $\Gamma $ in (\ref{KF-P}). Consider
the integral along contour $\Gamma _{1}.$ Then
\begin{equation*}
\left\vert \int\nolimits_{\Gamma _{1}}\tilde{P}\left( x,s\right) \mathrm{e}%
^{st}\mathrm{d}s\right\vert \leq \int_{0}^{s_{0}}\left\vert \tilde{P}\left(
x,p+\mathrm{i}R\right) \right\vert \left\vert \mathrm{e}^{\left( p+\mathrm{i}%
R\right) t}\right\vert \mathrm{d}p.
\end{equation*}%
Let $R\rightarrow \infty .$ In order to estimate $\left\vert \tilde{P}\left(
x,p\pm \mathrm{i}R\right) \right\vert ,$ using (iii) of proposition \ref%
{propP}, we write
\begin{eqnarray*}
&&M\left( p\pm \mathrm{i}R\right) \sim v\pm \mathrm{i}w, \\
&&v=\sqrt{\frac{a}{b}}\frac{1}{\ln \left( aR\right) }\frac{\ln \left(
aR\right) \ln \left( bR\right) }{\sqrt[4]{\left( \ln \left( aR\right) \ln
\left( bR\right) \right) ^{2}+\left( \frac{\pi }{2}\ln \frac{b}{a}\right)
^{2}}}, \\
&&w=-\sqrt{\frac{a}{b}}\frac{1}{\ln \left( aR\right) }\frac{\frac{\pi }{2}%
\ln \frac{b}{a}}{\sqrt[4]{\left( \ln \left( aR\right) \ln \left( bR\right)
\right) ^{2}+\left( \frac{\pi }{2}\ln \frac{b}{a}\right) ^{2}}}.
\end{eqnarray*}%
Then, as $R\rightarrow \infty ,$%
\begin{eqnarray}
\left\vert \tilde{P}\left( x,p\pm \mathrm{i}R\right) \right\vert &\sim
&\left\vert \frac{\sinh \left[ x\left( pv-Rw\right) \pm \mathrm{i}x\left(
pw+Rv\right) \right] }{\sinh \left[ \left( pv-Rw\right) \pm \mathrm{i}\left(
pw+Rv\right) \right] }\right\vert  \notag \\
&\leq &\frac{\mathrm{e}^{x\left( pv-Rw\right) }+\mathrm{e}^{-x\left(
pv-Rw\right) }}{\left\vert \mathrm{e}^{pv-Rw}-\mathrm{e}^{-\left(
pv-Rw\right) }\right\vert }  \notag \\
&=&\mathrm{e}^{-\left( 1-x\right) \left( pv-Rw\right) }\frac{1+\mathrm{e}%
^{-2x\left( pv-Rw\right) }}{\left\vert 1-\mathrm{e}^{-2\left( pv-Rw\right)
}\right\vert }\rightarrow 0.  \label{pt-je-nula}
\end{eqnarray}%
The previous statement is valid since, as$\;\;R\rightarrow \infty ,$
\begin{eqnarray*}
pv-Rw &=&\sqrt{\frac{a}{b}}\frac{1}{\ln \left( aR\right) }\frac{1}{\sqrt[4]{%
\left( \ln \left( aR\right) \ln \left( bR\right) \right) ^{2}+\left( \frac{%
\pi }{2}\ln \frac{b}{a}\right) ^{2}}} \\
&&\times \left( p\ln \left( aR\right) \ln \left( bR\right) +R\frac{\pi }{2}%
\ln \frac{b}{a}\right) \\
&\sim &\sqrt{\frac{a}{b}}\left( p\sqrt{\frac{\ln \left( bR\right) }{\ln
\left( aR\right) }}+\frac{\pi }{2}\ln \frac{b}{a}\frac{R}{\ln \left(
aR\right) \sqrt{\ln \left( aR\right) \ln \left( bR\right) }}\right)
\rightarrow \infty .
\end{eqnarray*}%
Therefore, according to (\ref{pt-je-nula}), we have%
\begin{equation*}
\lim_{R\rightarrow \infty }\left\vert \int\nolimits_{\Gamma _{1}}\tilde{P}%
\left( x,s\right) \mathrm{e}^{st}\mathrm{d}s\right\vert =0.
\end{equation*}%
By the use of (\ref{pt-je-nula}), we conclude that similar arguments are
valid for the integral along the contour $\Gamma _{6}.$ Thus,%
\begin{equation*}
\lim\limits_{R\rightarrow \infty }\left\vert \int\nolimits_{\Gamma _{6}}%
\tilde{P}\left( x,s\right) \mathrm{e}^{st}\mathrm{d}s\right\vert =0.
\end{equation*}%
Next, we consider the integral along contour $\Gamma _{2}$
\begin{equation*}
\left\vert \int\nolimits_{\Gamma _{2}}\tilde{P}\left( x,s\right) \mathrm{e}%
^{st}\mathrm{d}s\right\vert \leq \int\nolimits_{\frac{\pi }{2}}^{\pi
}R\left\vert \mathrm{e}^{R\left( 1-x\right) \mathrm{e}^{\mathrm{i}\phi
}M\left( R\mathrm{e}^{\mathrm{i}\phi }\right) }\right\vert \left\vert \frac{%
\mathrm{e}^{2xR\mathrm{e}^{\mathrm{i}\phi }M\left( R\mathrm{e}^{\mathrm{i}%
\phi }\right) }-1}{\mathrm{e}^{2R\mathrm{e}^{\mathrm{i}\phi }M\left( R%
\mathrm{e}^{\mathrm{i}\phi }\right) }-1}\right\vert \mathrm{e}^{Rt\cos \phi }%
\mathrm{d}\phi .
\end{equation*}%
Since $M\sim \sqrt{\frac{a}{b}}$ as $\left\vert s\right\vert \rightarrow
\infty $ and $\cos \phi \leq 0$ for $\phi \in \left[ \frac{\pi }{2},\pi %
\right] ,$ by the Lebesgue theorem, we have%
\begin{equation*}
\lim_{R\rightarrow \infty }\left\vert \int\nolimits_{\Gamma _{2}}\tilde{P}%
\left( x,s\right) \mathrm{e}^{st}\mathrm{d}s\right\vert \leq
\lim_{R\rightarrow \infty }\int\nolimits_{\frac{\pi }{2}}^{\pi }R\,\mathrm{e}%
^{R\cos \phi \left( t+\left( 1-x\right) \sqrt{\frac{a}{b}}\right) }\mathrm{d}%
\phi =0.
\end{equation*}%
Similar arguments are valid for the integral along the contour $\Gamma _{5}.$
Thus,%
\begin{equation*}
\lim\limits_{R\rightarrow \infty }\left\vert \int\nolimits_{\Gamma _{5}}%
\tilde{P}\left( x,s\right) \mathrm{e}^{st}\mathrm{d}s\right\vert =0.
\end{equation*}%
The integration along contour $\Gamma _{\varepsilon }$ gives%
\begin{eqnarray*}
\lim_{\varepsilon \rightarrow 0}\left\vert \int\nolimits_{\Gamma
_{\varepsilon }}\tilde{P}\left( x,s\right) \mathrm{e}^{st}\mathrm{d}%
s\right\vert &=&\lim_{\varepsilon \rightarrow 0}\int\nolimits_{\pi }^{-\pi
}\varepsilon \left\vert \mathrm{e}^{-\varepsilon \left( 1-x\right) \mathrm{e}%
^{\mathrm{i}\phi }M\left( \varepsilon \mathrm{e}^{\mathrm{i}\phi }\right)
}\right\vert \\
&&\times \left\vert \frac{1-\mathrm{e}^{-2x\varepsilon \mathrm{e}^{\mathrm{i}%
\phi }M\left( \varepsilon \mathrm{e}^{\mathrm{i}\phi }\right) }}{1-\mathrm{e}%
^{-2\varepsilon \mathrm{e}^{\mathrm{i}\phi }M\left( \varepsilon \mathrm{e}^{%
\mathrm{i}\phi }\right) }}\right\vert \mathrm{e}^{\varepsilon t\cos \phi }%
\mathrm{d}\phi
\begin{tabular}{l}
=%
\end{tabular}%
0.
\end{eqnarray*}%
Integrals along parts of contour $\Gamma _{3},$ $\Gamma _{4}$ and $\gamma
_{0}$ give%
\begin{eqnarray*}
\lim_{\substack{ R\rightarrow \infty  \\ \varepsilon \rightarrow 0}}%
\int\nolimits_{\Gamma _{3}}\tilde{P}\left( x,s\right) \mathrm{e}^{st}\mathrm{%
d}s &=&\int\nolimits_{0}^{\infty }\frac{\sinh \left( xqM\left( q\mathrm{e}^{%
\mathrm{i}\pi }\right) \right) }{\sinh \left( qM\left( q\mathrm{e}^{\mathrm{i%
}\pi }\right) \right) }\mathrm{e}^{-qt}\mathrm{d}q, \\
\lim_{\substack{ R\rightarrow \infty  \\ \varepsilon \rightarrow 0}}%
\int\nolimits_{\Gamma _{4}}\tilde{P}\left( x,s\right) \mathrm{e}^{st}\mathrm{%
d}s &=&-\int\nolimits_{0}^{\infty }\frac{\sinh \left( xqM\left( q\mathrm{e}%
^{-\mathrm{i}\pi }\right) \right) }{\sinh \left( qM\left( q\mathrm{e}^{-%
\mathrm{i}\pi }\right) \right) }\mathrm{e}^{-qt}\mathrm{d}q, \\
\lim_{R\rightarrow \infty }\int\nolimits_{\gamma _{0}}\tilde{P}\left(
x,s\right) \mathrm{e}^{st}\mathrm{d}s &=&2\pi \mathrm{i}P\left( x,t\right) .
\end{eqnarray*}%
Now, by the Cauchy residues theorem (\ref{KF-P}), the function $P$ is
determined by (\ref{P1}).

In order to see that $P$ is a real-valued function, we use $M\left( q\mathrm{%
e}^{\pm \mathrm{i}\pi }\right) ,$ $q\in \left[ 0,\infty \right) ,$ and note
that $M\left( q\mathrm{e}^{-\mathrm{i}\pi }\right) =\overline{M\left( q%
\mathrm{e}^{\mathrm{i}\pi }\right) },$ where the bar denotes the complex
conjugation. Due to the exponential in the hyperbolic sine, we have $\sinh
\left( xqM\left( q\mathrm{e}^{-\mathrm{i}\pi }\right) \right) =\overline{%
\sinh \left( xqM\left( q\mathrm{e}^{\mathrm{i}\pi }\right) \right) }$ and
therefore the integrand in (\ref{P1}) is of the form%
\begin{eqnarray*}
&&\frac{\sinh \left( xqM\left( q\mathrm{e}^{-\mathrm{i}\pi }\right) \right)
}{\sinh \left( qM\left( q\mathrm{e}^{-\mathrm{i}\pi }\right) \right) }-\frac{%
\sinh \left( xqM\left( q\mathrm{e}^{\mathrm{i}\pi }\right) \right) }{\sinh
\left( qM\left( q\mathrm{e}^{\mathrm{i}\pi }\right) \right) } \\
&&\;\;\;\;=\overline{\left( \frac{\sinh \left( xqM\left( q\mathrm{e}^{%
\mathrm{i}\pi }\right) \right) }{\sinh \left( qM\left( q\mathrm{e}^{\mathrm{i%
}\pi }\right) \right) }\right) }-\frac{\sinh \left( xqM\left( q\mathrm{e}^{%
\mathrm{i}\pi }\right) \right) }{\sinh \left( qM\left( q\mathrm{e}^{\mathrm{i%
}\pi }\right) \right) } \\
&&\;\;\;\;=-2\mathrm{i\func{Im}}\left( \frac{\sinh \left( xqM\left( q\mathrm{%
e}^{\mathrm{i}\pi }\right) \right) }{\sinh \left( qM\left( q\mathrm{e}^{%
\mathrm{i}\pi }\right) \right) }\right) ,
\end{eqnarray*}%
which implies that the first term in (\ref{P1}) is real. \newline
Next, we examine $\func{Res}\left( \tilde{P}\left( x,s\right) \mathrm{e}%
^{st},{}_{P}s_{n}^{\left( \pm \right) }\right) $ in order to prove that the
sum of residues is also real. By (\ref{rez}) and
\begin{eqnarray*}
&&\left[ 1-\frac{\ln \frac{b}{a}}{2\ln \left( as\right) \ln \left( bs\right)
}+\frac{\left( b-a\right) s}{2\left( as-1\right) \left( bs-1\right) }\right]
_{s={}_{P}s_{n}^{\left( -\right) }} \\
&&\;\;\;\;=\overline{\left( \left[ 1-\frac{\ln \frac{b}{a}}{2\ln \left(
as\right) \ln \left( bs\right) }+\frac{\left( b-a\right) s}{2\left(
as-1\right) \left( bs-1\right) }\right] _{s={}_{P}s_{n}^{\left( +\right)
}}\right) }
\end{eqnarray*}%
we obtain that $\func{Res}\left( \tilde{P}\left( x,s\right) \mathrm{e}%
^{st},{}_{P}s_{n}^{\left( -\right) }\right) =\overline{\func{Res}\left(
\tilde{P}\left( x,s\right) \mathrm{e}^{st},{}_{P}s_{n}^{\left( +\right)
}\right) }.$ It is clear that%
\begin{equation*}
\func{Res}\left( \tilde{P}\left( x,s\right) \mathrm{e}^{st},{}_{P}s_{n}^{%
\left( +\right) }\right) +\func{Res}\left( \tilde{P}\left( x,s\right)
\mathrm{e}^{st},{}_{P}s_{n}^{\left( -\right) }\right) =2\func{Re}\left(
\func{Res}\left( \tilde{P}\left( x,s\right) \mathrm{e}^{st},{}_{P}s_{n}^{%
\left( +\right) }\right) \right) .
\end{equation*}%
This implies that the second term in (\ref{P1}) is also real for $n\in
\mathbb{N}
.$ Hence, (\ref{P1}) is a real-valued function.

Let $t<0.$ We prove that the integral over $\gamma _{0}$ does not depend on
the choice of $s_{0}$ (see figure \ref{fig-1}). Let $\bar{\Gamma}=\gamma
_{0}\cup \gamma _{1}\cup \gamma _{0}^{\prime }\cup \gamma _{2}$ (see figure %
\ref{fig-1-1}), where $s_{0}$ and $s_{0}^{\prime }$ are chosen so that all
poles, i.e. solutions of (\ref{21-1}), lie on the left of $\gamma _{0}$. The
Cauchy residues theorem yields ($x\in \left[ 0,1\right] $)%
\begin{equation*}
\oint\nolimits_{\bar{\Gamma}}\tilde{P}\left( x,s\right) \mathrm{e}^{st}%
\mathrm{d}s=0.
\end{equation*}%
This and (\ref{pt-je-nula}) imply%
\begin{equation*}
\lim_{R\rightarrow \infty }\left\vert \int_{\gamma _{1}}\tilde{P}\left(
x,s\right) \mathrm{e}^{st}\mathrm{d}s\right\vert \leq \lim_{R\rightarrow
\infty }\int_{s_{0}}^{s_{0}^{\prime }}\left\vert \tilde{P}\left( x,v+\mathrm{%
i}R\right) \right\vert \left\vert \mathrm{e}^{\left( v+\mathrm{i}R\right)
t}\right\vert \mathrm{d}v=0.
\end{equation*}%
Similar arguments hold for the integral along $\gamma _{2}.$ Therefore, by
the Cauchy residues theorem, integrals along $\gamma _{0}$ and $\gamma
_{0}^{\prime }$ are equal and the inversion of the Laplace transformation
does not depend on the choice of $s_{0}$ as well as on the choice of $%
s_{0}^{\prime }.$

The Cauchy residues theorem yields ($x\in \left[ 0,1\right] $)
\begin{equation*}
\oint\nolimits_{\tilde{\Gamma}}\tilde{P}\left( x,s\right) \mathrm{e}^{st}%
\mathrm{d}s=0,
\end{equation*}%
where $\tilde{\Gamma}=\gamma _{0}\cup \Gamma _{r}$ (see figure \ref{fig-1-2}%
), with the assumption that all poles, i.e. solutions of (\ref{21-1}), lie
on the left of $\gamma _{0}$.
\begin{figure}[h]
\begin{minipage}{60mm}
\centering
\includegraphics[scale=0.45]{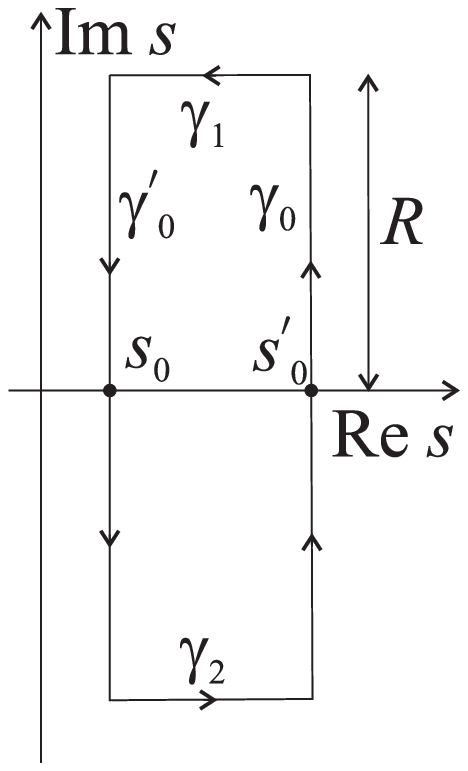}
\caption{Integration contour $\bar{\Gamma}$} \label{fig-1-1}
\label{fig:3}
\end{minipage}
\hfil
\begin{minipage}{60mm}
\centering
\includegraphics[scale=0.45]{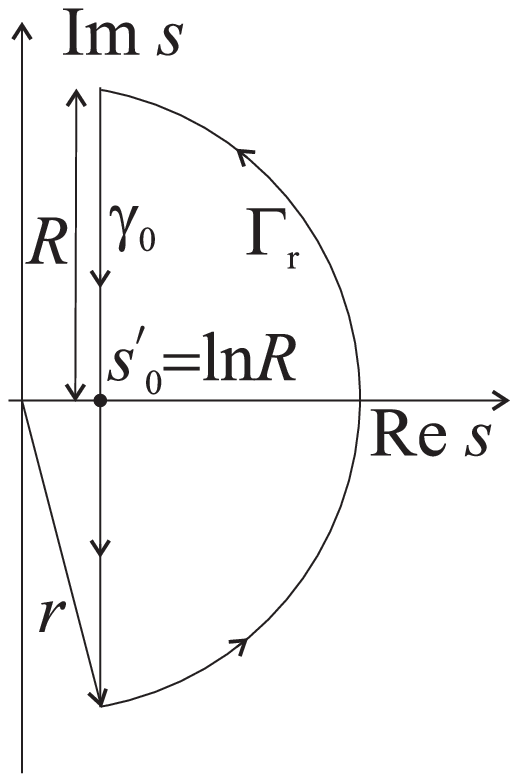}
\caption{Integration contour $\tilde{\Gamma}$} \label{fig-1-2}
\end{minipage}
\end{figure}
Let $r\left( R\right) =\sqrt{R^{2}+s_{0}^{2}}.$ Consider%
\begin{equation*}
\left\vert \int_{\Gamma _{r}}\tilde{P}\left( x,s\right) \mathrm{e}^{st}%
\mathrm{d}s\right\vert \leq \int_{-\phi _{0}\left( r\left( R\right) \right)
}^{\phi _{0}\left( r\left( R\right) \right) }\sqrt{R^{2}+s_{0}^{2}}%
\left\vert \tilde{P}\left( x,\sqrt{R^{2}+s_{0}^{2}}\mathrm{e}^{\mathrm{i}%
\phi }\right) \right\vert \mathrm{e}^{t\sqrt{R^{2}+s_{0}^{2}}\cos \phi }%
\mathrm{d}\phi ,
\end{equation*}%
where $\lim\limits_{R\rightarrow \infty }\phi _{0}\left( r\left( R\right)
\right) =\frac{\pi }{2}.$ Since $M\sim \sqrt{\frac{a}{b}}$ as $\left\vert
s\right\vert \rightarrow \infty ,$ by (\ref{4-1}) and (\ref{M}), we have, as
$\left\vert s\right\vert \rightarrow \infty ,$%
\begin{equation}
\left\vert \tilde{P}\left( x,s\right) \right\vert =\left\vert \mathrm{e}%
^{-\left( 1-x\right) sM\left( s\right) }\frac{1-\mathrm{e}^{-2xsM\left(
s\right) }}{1-\mathrm{e}^{-2sM\left( s\right) }}\right\vert \leq C,\;\;x\in %
\left[ 0,1\right] ,\;s\in
\mathbb{C}
\backslash \left( -\infty ,0\right] .  \label{p-besk}
\end{equation}%
By (\ref{p-besk}), we have%
\begin{equation*}
\lim_{R\rightarrow \infty }\left\vert \int_{\Gamma _{r}}\tilde{P}\left(
x,s\right) \mathrm{e}^{st}\mathrm{d}s\right\vert \leq C\lim_{R\rightarrow
\infty }\int_{-\phi _{0}\left( r\left( R\right) \right) }^{\phi _{0}\left(
r\left( R\right) \right) }\sqrt{R^{2}+s_{0}^{2}}\mathrm{e}^{t\sqrt{%
R^{2}+s_{0}^{2}}\cos \phi }\mathrm{d}\phi =0,
\end{equation*}%
since $t<0$ and $\cos \phi >0.$ Therefore, we proved (\ref{uhtm0}).

By the use of (\ref{P1}) and (\ref{uhtm0}) in (\ref{u-h}) and by calculating
the convolution we obtain (\ref{uha}) and (\ref{uha0}).

In order to prove that $u_{H}$ is a continuous function at $t=0$, we will
use Lebesgue dominated convergence theorem. Let
\begin{equation*}
f\left( q,x\right) :=\frac{\Upsilon _{0}}{2\pi \mathrm{i}}\left( \frac{\sinh
\left( xqM\left( q\mathrm{e}^{-\mathrm{i}\pi }\right) \right) }{\sinh \left(
qM\left( q\mathrm{e}^{-\mathrm{i}\pi }\right) \right) }-\frac{\sinh \left(
xqM\left( q\mathrm{e}^{\mathrm{i}\pi }\right) \right) }{\sinh \left(
qM\left( q\mathrm{e}^{\mathrm{i}\pi }\right) \right) }\right) ,\;\;q\in
\left( 0,\infty \right) ,\;x\in \left[ 0,1\right] .
\end{equation*}%
Then
\begin{equation}
\left\vert \dint\nolimits_{0}^{\infty }f\left( q,x\right) \frac{1-\mathrm{e}%
^{-qt}}{q}\mathrm{d}q\right\vert \rightarrow 0\;\;\text{as}\;\;t\rightarrow
0.  \label{zvez}
\end{equation}%
By simple calculations we have $\frac{1-\mathrm{e}^{-qt}}{q}\leq Ct$ if $%
0<q<1$ and $\frac{1-\mathrm{e}^{-qt}}{q}\leq 1-\mathrm{e}^{-qt}$ if $q\geq
1. $ Thus%
\begin{equation*}
f\left( q,x\right) \frac{1-\mathrm{e}^{-qt}}{q}\leq Cf\left( q,x\right)
,\;\;q>0
\end{equation*}%
and since $\frac{1-\mathrm{e}^{-qt}}{q}\rightarrow 0$ as$\ t\rightarrow 0,$ (%
\ref{zvez}) follows.

In proving the continuity of $u_{H}$ at $t=0,$ by (\ref{u-h}), we estimated%
\begin{equation*}
\int_{0}^{t}P\left( x,\tau \right) \mathrm{d}\tau ,\;\;x\in \left[ 0,1\right]
,\;t>0,
\end{equation*}%
and actually proved that $P$ is integrable function on any interval $\left[
0,T\right] ,$ $T>0.$ Thus, $P$ is locally integrable on $%
\mathbb{R}
.$
\end{proof}

\subsubsection{Case $\Upsilon =\Upsilon _{0}H+F$\label{hplusnesto}}

\begin{condition}
\label{cond}Let $F$ be a locally integrable function, equal to zero for $%
t\leq 0,$ such that its Laplace transformation exists in $%
\mathbb{C}
\backslash \left( -\infty ,0\right] .$ Assume:

\begin{enumerate}
\item[(i)] $\tilde{F}$ is analytic and $\tilde{F}\neq 0$ in $%
\mathbb{C}
\backslash \left( -\infty ,0\right] ;$

\item[(ii)] for some $\alpha >1,$ $\tilde{F}\left( s\right) \sim \frac{1}{%
\left\vert s\right\vert ^{\alpha }},$ $s\in
\mathbb{C}
\backslash \left( -\infty ,0\right] ,$ as $\left\vert s\right\vert
\rightarrow \infty ;$

\item[(iii)] $s\tilde{F}\left( s\right) \sim o\left( 1\right) ,$ $s\in
\mathbb{C}
\backslash \left( -\infty ,0\right] ,$ as $\left\vert s\right\vert
\rightarrow 0.$
\end{enumerate}
\end{condition}

If the boundary condition (\ref{BC-u-bd}) is given by (\ref{6}), then the
solution to (\ref{sys-bd}), (\ref{IC-bd}), (\ref{BC-u-bd}), given by (\ref%
{4-0}) in the Laplace domain, reads formally%
\begin{equation*}
\tilde{u}\left( x,s\right) =\tilde{u}_{H}\left( x,s\right) +\tilde{F}\left(
s\right) \tilde{P}\left( x,s\right) ,\;\;x\in \left[ 0,1\right] ,\;s\in
\mathbb{C}
\backslash \left( -\infty ,0\right] ,
\end{equation*}%
and in the time domain it is%
\begin{equation*}
u\left( x,t\right) =u_{H}\left( x,t\right) +F\left( t\right) \ast P\left(
x,t\right) ,\;\;x\in \left[ 0,1\right] ,\;t\in
\mathbb{R}
.
\end{equation*}%
The existence of $u_{H}$ is shown in \S \ref{H}, therefore it remains to
show the existence of%
\begin{equation*}
u_{F}\left( x,t\right) =F\left( t\right) \ast P\left( x,t\right) ,\;\;x\in
\left[ 0,1\right] ,\;t\in
\mathbb{R}
.
\end{equation*}

Let $t>0.$ The Cauchy residues theorem yields ($x\in \left[ 0,1\right] $)
\begin{equation}
\oint\nolimits_{\Gamma }\tilde{u}_{F}\left( x,s\right) \mathrm{e}^{st}%
\mathrm{d}s=2\pi \mathrm{i}\sum_{n=1}^{\infty }\left[ \func{Res}\left(
\tilde{u}_{F}\left( x,s\right) \mathrm{e}^{st},{}_{P}s_{n}^{\left( +\right)
}\right) +\func{Res}\left( \tilde{u}_{F}\left( x,s\right) \mathrm{e}%
^{st},{}_{P}s_{n}^{\left( -\right) }\right) \right] ,  \label{KF-PF}
\end{equation}%
where $\Gamma =\Gamma _{1}\cup \Gamma _{2}\cup \Gamma _{3}\cup \Gamma
_{\varepsilon }\cup \Gamma _{4}\cup \Gamma _{5}\cup \Gamma _{6}\cup \gamma
_{0}$ (see figure \ref{fig-1}). Since poles ${}_{P}s_{n}^{\left( \pm \right)
}$ of $\tilde{u}_{F}$ are actually the poles of $\tilde{P},$ that are
obtained from (\ref{21-1}) and they are simple for $n>n_{0},$ the residues
in (\ref{KF-PF}) can be calculated as%
\begin{equation}
\func{Res}\left( \tilde{u}_{F}\left( x,s\right) \mathrm{e}%
^{st},{}_{P}s_{n}^{\left( \pm \right) }\right) =\left[ \tilde{F}\left(
s\right) \frac{\sinh \left( xsM\left( s\right) \right) }{\frac{\mathrm{d}}{%
\mathrm{d}s}\left[ \sinh \left( sM\left( s\right) \right) \right] }\mathrm{e}%
^{st}\right] _{s={}_{P}s_{n}^{\left( \pm \right) }}.  \label{res-PF}
\end{equation}%
The proof that the sum in (\ref{KF-PF}) converges is analog to the one
presented in \S \ref{H}.

Consider the integral along contour $\Gamma _{1}.$ It reads%
\begin{equation*}
\left\vert \int\nolimits_{\Gamma _{1}}\tilde{u}_{F}\left( x,s\right) \mathrm{%
e}^{st}\mathrm{d}s\right\vert \leq \int_{0}^{s_{0}}\left\vert \tilde{F}%
\left( p+\mathrm{i}R\right) \right\vert \left\vert \tilde{P}\left( x,p+%
\mathrm{i}R\right) \right\vert \left\vert \mathrm{e}^{\left( p+\mathrm{i}%
R\right) t}\right\vert \mathrm{d}p.
\end{equation*}%
According to (\ref{p-besk}) and condition \ref{cond}, we have%
\begin{equation*}
\lim_{R\rightarrow \infty }\left\vert \int\nolimits_{\Gamma _{1}}\tilde{u}%
_{F}\left( x,s\right) \mathrm{e}^{st}\mathrm{d}s\right\vert \leq
C\lim_{R\rightarrow \infty }\int_{0}^{s_{0}}\frac{1}{\left( \sqrt{p^{2}+R^{2}%
}\right) ^{\alpha }}\mathrm{e}^{pt}\mathrm{d}p=0.
\end{equation*}%
The integral along contour $\Gamma _{2}$ reads
\begin{eqnarray*}
\left\vert \int\nolimits_{\Gamma _{2}}\tilde{u}_{F}\left( x,s\right) \mathrm{%
e}^{st}\mathrm{d}s\right\vert &\leq &\int\nolimits_{\frac{\pi }{2}}^{\pi
}\left\vert \tilde{F}\left( R\mathrm{e}^{\mathrm{i}\phi }\right) \right\vert
\left\vert \mathrm{e}^{R\left( 1-x\right) \mathrm{e}^{\mathrm{i}\phi
}M\left( R\mathrm{e}^{\mathrm{i}\phi }\right) }\right\vert \\
&&\times \left\vert \frac{\mathrm{e}^{2xR\mathrm{e}^{\mathrm{i}\phi }M\left(
R\mathrm{e}^{\mathrm{i}\phi }\right) }-1}{\mathrm{e}^{2R\mathrm{e}^{\mathrm{i%
}\phi }M\left( R\mathrm{e}^{\mathrm{i}\phi }\right) }-1}\right\vert \mathrm{e%
}^{Rt\cos \phi }R\mathrm{d}\phi .
\end{eqnarray*}%
In order to apply the Lebesgue theorem, we need $\alpha >1$ in condition \ref%
{cond} (actually it is enough to have $\alpha \geq 1,$ but case $\alpha =1$
is already considered). Since $M\sim \sqrt{\frac{a}{b}}$ as $\left\vert
s\right\vert \rightarrow \infty $ and $\cos \phi \leq 0$ for $\phi \in \left[
\frac{\pi }{2},\pi \right] ,$ we have%
\begin{equation*}
\lim_{R\rightarrow \infty }\left\vert \int\nolimits_{\Gamma _{2}}\tilde{u}%
_{F}\left( x,s\right) \mathrm{e}^{st}\mathrm{d}s\right\vert \leq
C\lim_{R\rightarrow \infty }\int\nolimits_{\frac{\pi }{2}}^{\pi }R^{1-\alpha
}\mathrm{e}^{R\cos \phi \left( t+\left( 1-x\right) \sqrt{\frac{a}{b}}\right)
}\mathrm{d}\phi =0.
\end{equation*}%
Similar arguments are valid for the integral along the contour $\Gamma _{5}.$
Thus,%
\begin{equation*}
\lim\limits_{R\rightarrow \infty }\left\vert \int\nolimits_{\Gamma _{5}}%
\tilde{u}_{F}\left( x,s\right) \mathrm{e}^{st}\mathrm{d}s\right\vert =0.
\end{equation*}%
The integration along contour $\Gamma _{\varepsilon }$ gives%
\begin{eqnarray*}
\left\vert \int\nolimits_{\Gamma _{\varepsilon }}\tilde{u}_{F}\left(
x,s\right) \mathrm{e}^{st}\mathrm{d}s\right\vert &\leq &\int\nolimits_{-\pi
}^{\pi }\left\vert \tilde{F}\left( \varepsilon \mathrm{e}^{\mathrm{i}\phi
}\right) \right\vert \left\vert \mathrm{e}^{-\varepsilon \left( 1-x\right)
\mathrm{e}^{\mathrm{i}\phi }M\left( \varepsilon \mathrm{e}^{\mathrm{i}\phi
}\right) }\right\vert \\
&&\times \left\vert \frac{1-\mathrm{e}^{-2x\varepsilon \mathrm{e}^{\mathrm{i}%
\phi }M\left( \varepsilon \mathrm{e}^{\mathrm{i}\phi }\right) }}{1-\mathrm{e}%
^{-2\varepsilon \mathrm{e}^{\mathrm{i}\phi }M\left( \varepsilon \mathrm{e}^{%
\mathrm{i}\phi }\right) }}\right\vert \mathrm{e}^{\varepsilon t\cos \phi
}\varepsilon \mathrm{d}\phi ,
\end{eqnarray*}%
and this tends to zero as $\varepsilon \rightarrow 0,$ according to
condition \ref{cond}. Integrals along parts of contour $\Gamma _{3},$ $%
\Gamma _{4}$ and $\gamma _{0}$ give%
\begin{eqnarray*}
\lim_{\substack{ R\rightarrow \infty  \\ \varepsilon \rightarrow 0}}%
\int\nolimits_{\Gamma _{3}}\tilde{u}_{F}\left( x,s\right) \mathrm{e}^{st}%
\mathrm{d}s &=&\int\nolimits_{0}^{\infty }\tilde{F}\left( q\mathrm{e}^{%
\mathrm{i}\pi }\right) \frac{\sinh \left( xqM\left( q\mathrm{e}^{\mathrm{i}%
\pi }\right) \right) }{\sinh \left( qM\left( q\mathrm{e}^{\mathrm{i}\pi
}\right) \right) }\mathrm{e}^{-qt}\mathrm{d}q, \\
\lim_{\substack{ R\rightarrow \infty  \\ \varepsilon \rightarrow 0}}%
\int\nolimits_{\Gamma _{4}}\tilde{u}_{F}\left( x,s\right) \mathrm{e}^{st}%
\mathrm{d}s &=&-\int\nolimits_{0}^{\infty }\tilde{F}\left( q\mathrm{e}^{-%
\mathrm{i}\pi }\right) \frac{\sinh \left( xqM\left( q\mathrm{e}^{-\mathrm{i}%
\pi }\right) \right) }{\sinh \left( qM\left( q\mathrm{e}^{-\mathrm{i}\pi
}\right) \right) }\mathrm{e}^{-qt}\mathrm{d}q, \\
\lim_{R\rightarrow \infty }\int\nolimits_{\gamma _{0}}\tilde{u}_{F}\left(
x,s\right) \mathrm{e}^{st}\mathrm{d}s &=&2\pi \mathrm{i}u_{F}\left(
x,t\right) .
\end{eqnarray*}%
Now, by the Cauchy residues theorem (\ref{KF-PF}), $u_{F}$ is determined as%
\begin{eqnarray}
u_{F}\left( x,t\right) &=&\frac{1}{2\pi \mathrm{i}}\dint\nolimits_{0}^{%
\infty }\left( \tilde{F}\left( q\mathrm{e}^{-\mathrm{i}\pi }\right) \frac{%
\sinh \left( xqM\left( q\mathrm{e}^{-\mathrm{i}\pi }\right) \right) }{\sinh
\left( qM\left( q\mathrm{e}^{-\mathrm{i}\pi }\right) \right) }\right.  \notag
\\
&&\left. -\tilde{F}\left( q\mathrm{e}^{\mathrm{i}\pi }\right) \frac{\sinh
\left( xM\left( q\mathrm{e}^{\mathrm{i}\pi }\right) \right) }{\sinh \left(
qM\left( q\mathrm{e}^{\mathrm{i}\pi }\right) \right) }\right) \mathrm{e}%
^{-qt}\mathrm{d}q  \notag \\
&&+\sum_{n=1}^{\infty }\left[ \func{Res}\left( \tilde{u}_{F}\mathrm{e}%
^{st},{}_{P}s_{n}^{\left( +\right) }\right) +\func{Res}\left( \tilde{u}_{F}%
\mathrm{e}^{st},{}_{P}s_{n}^{\left( -\right) }\right) \right] ,\;\;x\in %
\left[ 0,1\right] ,\;t>0,  \notag \\
&&  \label{u-f} \\
u_{F}\left( x,t\right) &=&0,\;\;x\in \left[ 0,1\right] ,\;t<0,  \notag
\end{eqnarray}%
where the residues are given by (\ref{res-PF}). Note that $u_{F}$ is a
locally integrable, real-valued function, which can shown similarly as in \S %
\ref{H}.

Therefore, in the case when boundary condition takes the form (\ref{6}), the
solution to system (\ref{sys-bd}), (\ref{IC-bd}), (\ref{BC-u-bd}) reads%
\begin{equation*}
u\left( x,t\right) =u_{H}\left( x,t\right) +u_{F}\left( x,t\right) ,\;\;x\in
\left[ 0,1\right] ,\;t>0,
\end{equation*}%
where $u_{H}\ $and $u_{F}$ are given by (\ref{uha}) and (\ref{u-f}),
respectively. Note that $u_{H}\ $and $u_{F}$ are equal to zero for $t<0.$
Again, we have that $u$ is a smooth function for $x\in \left[ 0,1\right] ,$\
$t>0.$

\subsection{Determination of stress $\protect\sigma $ in case of stress
relaxation\label{dts}}

We see that $\tilde{T},$ given by (\ref{7-1}), has the branch point at $s=0$
and poles at the same points as $\tilde{P}.$ Therefore, the poles of $\tilde{%
T}$ are given as solutions to (\ref{21-1}). Using the Cauchy residues theorem%
\begin{equation}
\oint\nolimits_{\Gamma }\tilde{T}\left( x,s\right) \mathrm{e}^{st}\mathrm{d}%
s=2\pi \mathrm{i}\sum_{n=1}^{\infty }\left[ \func{Res}\left( \tilde{T}\left(
x,s\right) \mathrm{e}^{st},{}_{P}s_{n}^{\left( +\right) }\right) +\func{Res}%
\left( \tilde{T}\left( x,s\right) \mathrm{e}^{st},{}_{P}s_{n}^{\left(
-\right) }\right) \right] ,  \label{KF-T}
\end{equation}%
where contour $\Gamma $ is given in figure \ref{fig-1}, we obtain $T$ in the
following way. Residues in (\ref{KF-T}) are given by%
\begin{equation}
\func{Res}\left( \tilde{T}\left( x,s\right) \mathrm{e}^{st},{}_{P}s_{n}^{%
\left( \pm \right) }\right) =\left[ \frac{\cosh \left( xsM\left( s\right)
\right) }{M\left( s\right) \frac{\mathrm{d}}{\mathrm{d}s}\left[ \sinh \left(
sM\left( s\right) \right) \right] }\mathrm{e}^{st}\right] _{s={}_{P}s_{n}^{%
\left( \pm \right) }},  \label{res-T}
\end{equation}%
where ${}_{P}s_{n}^{\left( \pm \right) },$ $n\in
\mathbb{N}
,$ are solutions of (\ref{21-1}).

Evaluating the integral at the left hand side of (\ref{KF-T}) in the same
way as in \S \ref{H}, we obtain%
\begin{eqnarray*}
T\left( x,t\right) &=&1+\frac{1}{2\pi \mathrm{i}}\dint\nolimits_{0}^{\infty
}\left( \frac{\cosh \left( xqM\left( q\mathrm{e}^{\mathrm{i}\pi }\right)
\right) }{M\left( q\mathrm{e}^{\mathrm{i}\pi }\right) \sinh \left( qM\left( q%
\mathrm{e}^{\mathrm{i}\pi }\right) \right) }\right. \\
&&\left. -\frac{\cosh \left( xqM\left( q\mathrm{e}^{-\mathrm{i}\pi }\right)
\right) }{M\left( q\mathrm{e}^{-\mathrm{i}\pi }\right) \sinh \left( qM\left(
q\mathrm{e}^{-\mathrm{i}\pi }\right) \right) }\right) \mathrm{e}^{-qt}%
\mathrm{d}q \\
&&+\sum_{n=1}^{\infty }\left[ \func{Res}\left( \tilde{T}\left( x,s\right)
\mathrm{e}^{st},{}_{P}s_{n}^{\left( +\right) }\right) \right. \\
&&\left. +\func{Res}\left( \tilde{T}\left( x,s\right) \mathrm{e}%
^{st},{}_{P}s_{n}^{\left( -\right) }\right) \right] ,\;\;x%
\begin{tabular}{l}
$\in $%
\end{tabular}%
\left[ 0,1\right] ,\;t%
\begin{tabular}{l}
\TEXTsymbol{>}%
\end{tabular}%
0, \\
T\left( x,t\right) &=&0,\;\;x\in \left[ 0,1\right] ,\;t<0,
\end{eqnarray*}%
where the residues are given by (\ref{res-T}). The proof is analogue to the
one presented in \S \ref{H}.

Thus, by (\ref{11}), we have
\begin{equation}
\sigma _{H}\left( x,t\right) =\Upsilon _{0}T\left( x,t\right) ,\;\;x\in
\left[ 0,1\right] ,\;t>0,  \label{sigmaH}
\end{equation}%
if boundary conditions (\ref{BC-u-bd}) are given by (\ref{5}). Note that $%
\sigma _{H}$ is a locally integrable function with the jump at $t=0$ and
smooth for $t>0$. Also in case when boundary conditions (\ref{BC-u-bd}) are
given by (\ref{6}), we have ($x\in \left[ 0,1\right] $)%
\begin{eqnarray*}
\sigma _{F}\left( x,t\right) &=&\sigma _{H}\left( x,t\right) +\frac{\mathrm{d%
}}{\mathrm{d}t}%
\big(%
F\left( t\right) \ast T\left( x,t\right)
\big)%
,\;\;t>0 \\
\sigma _{F}\left( x,t\right) &=&0,\;\;t<0.
\end{eqnarray*}%
This is a smooth function for $t>0.$ Note that $\sigma _{H}$ and $\sigma
_{F} $ are real-valued functions, which can shown similarly as in \S \ref{H}.

\section{Numerical examples for displacement and stress in case of stress
relaxation\label{ne}}

In this section, we give several numerical examples of displacement $u_{H}$
and stress $\sigma _{H},$ given by (\ref{uha}) and (\ref{sigmaH}),
respectively. In figure \ref{fig-2}, we show displacements, determined
according to (\ref{uha}), for two different positions. Parameters in (\ref%
{uha}) are chosen as follows: $\Upsilon _{0}=1,$ $a=0.045,$ $b=0.5.$ The
integration goes to $1000,$ while the number of residues in the sum is $400.$
\begin{figure}[h]
\centering
\includegraphics[scale=0.7]{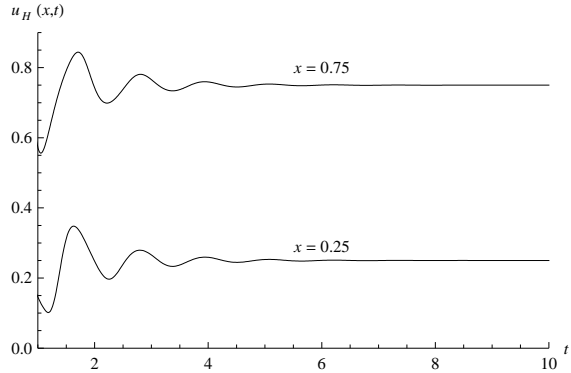}
\caption{Displacements $u_{H}\left( x,t\right) $ in stress relaxation
experiment, as a function of time at $x=0.25,$ $x=0.75$ for $t\in \left(
1,10\right) .$}
\label{fig-2}
\end{figure}

In figure \ref{fig-3}, we show the stresses determined according to (\ref%
{sigmaH}) for the same values of parameters used for figure \ref{fig-2}. In
order to emphasize stress relaxation process, we show stresses only for $%
t\geq 1$.
\begin{figure}[h]
\centering
\includegraphics[scale=0.7]{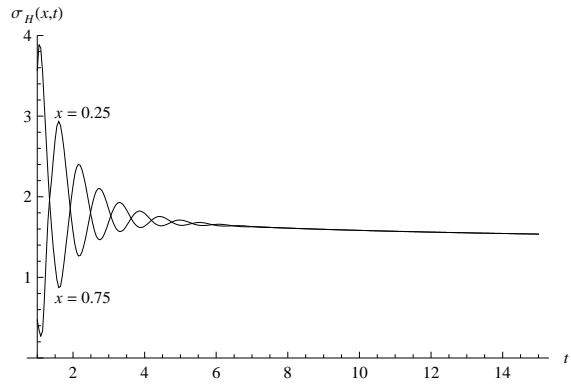}
\caption{Stresses $\protect\sigma _{H}\left( x,t\right) $ in stress
relaxation experiment, as a function of time at $x=0.25,$ $x=0.75$ for $t\in
\left( 1,15\right) .$}
\label{fig-3}
\end{figure}
From figure \ref{fig-3}, it is evident that the stress tends to a constant
value $\Upsilon _{0},$ at each point $x\in \left[ 0,1\right] .$

\section{Conclusion\label{conc}}

In this work we analyze displacements $u$ and stresses $\sigma $ for a
viscoelastic rod of finite length, which satisfy a constitutive equation of
distributed fractional order (\ref{3-1}), (\ref{a-b}). Displacement of a
free end of a rod is assumed to be $\Upsilon =\Upsilon _{0}H$ ($\Upsilon
_{0}>0$) and the displacement $u_{H}$ is obtained in the form (\ref{uha}).
Results for displacements and stresses are shown in figures \ref{fig-2} and %
\ref{fig-3}. Figures show oscillatory character of both stresses and
displacements. Oscillations are damped and for large time displacements show
linear dependence on the distance of a particle from a fixed end, while
stresses are approaching to the limiting value independently of the position
of the particle. For large time stress relaxation curves tend to curves
corresponding to quasistatic analysis (see work by \cite{a-2002} and \cite%
{drozd}). In figures \ref{fig-2} and \ref{fig-3}, we show displacements and
stresses for $t\geq 1$.

Finally, let us comment the choice of parameters in constitutive equation (%
\ref{3-1}), (\ref{a-b}). Our choice $a<b$ in (\ref{a-b}) is a result of the
requirement that entropy inequality is satisfied for each $\alpha \in \left(
0,1\right) $ (see a paper by \cite{b-t}). The constitutive equation (\ref%
{3-1}), (\ref{a-b}) is of a viscoelastic type. A generalization of the
problem treated here would include effects of viscoinertial type. In that
case (\ref{3-1}) would be replaced by
\begin{equation}
\int_{0}^{2}\phi _{1}\left( \alpha \right) {}_{0}D_{t}^{\alpha }\sigma
\left( x,t\right) \mathrm{d}\alpha =E\int_{0}^{2}\phi _{2}\left( \alpha
\right) {}_{0}D_{t}^{\alpha }\mathcal{E}\left( x,t\right) \mathrm{d}\alpha
,\;\;x\in \left[ 0,L\right] ,\;t>0.  \label{visc-inert}
\end{equation}%
For the analysis of a system (\ref{1}), (\ref{3}), (\ref{visc-inert}), one
could use the type of analysis presented in work by \citeasnoun{APZ-1} and %
\citeasnoun{APZ-2}.

\begin{ack}
This research was supported by Ministry of Science projects 144019 (T.M.A.
and D.Z.) and 144016 (S.P.).
\end{ack}

\end{document}